\documentclass[11pt]{article}
\pdfoutput=1
\usepackage{multirow}
\usepackage{setspace}
\usepackage{makecell}
\usepackage{fullpage,bm,amsmath,amssymb,color,amsthm}
\usepackage{algorithm}
\usepackage[noend]{algpseudocode}
\usepackage{natbib}
\setcitestyle{numbers,square}
\usepackage{hyperref}
\usepackage{nameref}
\usepackage{url}
\usepackage{soul}
\usepackage[normalem]{ulem}
\usepackage{array}
\usepackage{hhline}
\usepackage{authblk}
\usepackage{graphicx}
\usepackage{subfig}
\usepackage{epstopdf}
\graphicspath{{./}{./figures/}{./figures/Synthetic/}{./figures/FB/}}

\makeatletter
\def\BState{\State\hskip-\ALG@thistlm}
\makeatother

\newtheorem{theorem}{Theorem}[section]
\newtheorem{lemma}[theorem]{Lemma}

\newcommand\norm[1]{\left\lVert#1\right\rVert}

\allowdisplaybreaks
\newcommand\numberthis{\addtocounter{equation}{1}\tag{\theequation}}

\def\R{\mathbb{R}}

\def\BL{\bm{L}}

\def\BD{\bm{D}}
\def\BS{\bm{S}}

\def\lb{\left(}
\def\rb{\right)}

\def\ln{\left\|}
\def\rn{\right\|}

\def\lab{\left|}
\def\rab{\right|}

\newcommand{\KW}[1]{}

\begin{document}

\title{Accelerated Alternating Projections for Robust Principal Component Analysis}
\author[1]{\rm HanQin Cai}
\author[2]{\rm Jian-Feng Cai}
\author[3]{\rm Ke Wei}
\affil[1]{Department of Mathematics,
       University of California, Los Angeles\\
       
       Los Angeles, California, USA}
\affil[2]{Department of Mathematics, Hong Kong University of Science and Technology\\

Clear Water Bay, Kowloon, Hong Kong SAR, China}
\affil[3]{School of Data Science, Fudan University, Shanghai, China}
\affil[ ]{\text {$^1$hqcai@math.ucla.edu \qquad \qquad $^2$jfcai@ust.hk \qquad \qquad $^3$kewei@fudan.edu.cn}}
\maketitle

\begin{abstract}
We study robust PCA for the fully observed setting, which is about separating a low rank matrix $\BL$ and a sparse matrix $\BS$ from their sum $\BD=\BL+\BS$. In this paper, a new algorithm, dubbed accelerated alternating projections, is introduced  for robust PCA 
which significantly improves the computational efficiency of the existing alternating projections proposed in \citep{netrapalli2014non} when updating the low rank factor. The acceleration is achieved by first projecting a matrix onto some low dimensional subspace before obtaining a new estimate of the low rank matrix via truncated SVD.
Exact recovery guarantee has been established which shows linear convergence of the proposed algorithm. Empirical performance evaluations establish the advantage of our algorithm over other state-of-the-art algorithms for robust PCA.\end{abstract}


\section{Introduction}
Robust principal component analysis (RPCA) appears in a wide range of applications, including video and voice background subtraction  \citep{li2004statistical,huang2012singing}, sparse graphs clustering \citep{chen2012clustering}, 3D reconstruction \citep{mobahi2011holistic}, and fault isolation \citep{tharrault2008fault}. Suppose we are given a sum of a 
low rank matrix and a sparse matrix, denoted $\BD=\BL+\BS$. The goal of RPCA is to reconstruct $\BL$ and $\BS$ simultaneously from $\BD$. As a concrete example, for  foreground-background separation in video processing, $\bm{L}$ represents  static background through all the frames of a video which should be low rank while 
$\bm{S}$ represents moving objects which can be assumed to be sparse since typically they will not block a large portion of the screen for a long time. 

RPCA can be achieved by seeking a low rank matrix $\bm{L'}$ and a sparse matrix $\bm{S'}$ such that their sum fits the measurement matrix $\BD$ as well as possible:
 \begin{equation} \label{eq:non-convex model 2}
\min_{\bm{L'},\bm{S'}\in\mathbb{R}^{m\times n}} \|\bm{D}-\bm{L'}-\bm{S'}\|_F \quad \textnormal{ subject to } {rank}(\bm{L'}) \leq r \textnormal{ and } \|\bm{S'}\|_0 \leq |\Omega|,
\end{equation}
where $r$ denotes the rank of the underlying low rank matrix, $\Omega$ denotes the support set of the underlying sparse matrix, and $\|\bm{S'}\|_0$ counts the number 
of non-zero entries in $\bm{S'}$.
Compared to the traditional principal component analysis (PCA) which computes a low rank approximation of a data matrix,   RPCA is less sensitive to outliers since it includes a sparse part in the formulation. 
Since the seminal works of \citep{wright2009robust,candes2011robust,chandrasekaran2011rank}, RPCA has received intensive investigations both from theoretical and algorithmic aspects.  Noticing that 
\eqref{eq:non-convex model 2} is a non-convex problem, some of the earlier works  focus on the following convex relaxation of RPCA: 
\begin{equation} \label{eq:convex model}
\min_{\bm{L'},\bm{S'}\in\mathbb{R}^{m\times n}} \|\bm{L'}\|_*+\lambda\|\bm{S'}\|_1 \quad \textnormal{ subject to } \bm{L'} + \bm{S'}=\bm{D},
\end{equation}
where $\|\cdot\|_*$ is the nuclear norm ({\em viz.} trace norm) of matrices, $\lambda$ is the regularization parameter, and $\|\cdot\|_1$ denotes the $\ell_1$-norm of the vectors obtained by stacking the columns of associated matrices. 
Under some mild conditions, it has been proven that the RPCA problem can be solved exactly by the aforementioned  convex relaxation \cite{candes2011robust,chandrasekaran2011rank}.
However, a limitation of the convex relaxation based approach is that the resulting semidefinite programming is computationally rather expensive to solve, even for medium size matrices. Alternative to the convex relaxation, many non-convex algorithms have been designed to target \eqref{eq:non-convex model 2} directly. This line of research will be reviewed in more detail in Section~\ref{subsec:related work} after our approach has been introduced. 

This paper targets the non-convex optimization for RPCA directly. The main contributions of this work are two-fold. Firstly, we propose a new algorithm,  accelerated alternating projections (AccAltProj), for RPCA, which is substantially faster than  other  state-of-the-art algorithms. Secondly,  exact recovery of accelerated alternating projections has been established for the fixed sparsity model, where we assume the ratio of the number of non-zero entries in each row and column of $\BS$ is less than a threshold. 

\subsection{Assumptions}
It is clear that  the RPCA problem is ill-posed without any additional conditions. Common assumptions are that $\BL$ cannot be too sparse and $\BS$ cannot be locally too dense, which are formalized in \nameref{assume:Inco} and \nameref{assume:Sparse}, respectively. 
\paragraph*{A1}\label{assume:Inco} 
\textit{The underlying low rank matrix $\bm{L}\in \mathbb{R}^{m\times n}$ is a rank-$r$ matrix with ${\mu}$-incoherence, that is
\begin{equation*}
\max_i \|\bm{e}_i^T \bm{U}\|_2\leq \sqrt{\frac{\mu r}{m}}, \quad\textnormal{and}\quad \max_j \|\bm{e}_j^T \bm{V}\|_2\leq \sqrt{\frac{\mu r}{n}}
\end{equation*}
hold for a positive numerical constant $1\leq\mu\leq\frac{\min\{m,n\}}{r}$, where $\bm{L}=\bm{U}\bm{\Sigma} \bm{V}^T$ is the SVD of $\bm{L}$. }

Assumption \nameref{assume:Inco} was first introduced in \citep{candes2009exact} for low rank matrix completion, and now it is a very standard assumption for related low rank reconstruction problems. It basically states that the left and right singular vectors of $\BL$ are weakly correlated with the canonical basis, which implies $\BL$ cannot be a very sparse matrix.

\paragraph*{A2}\label{assume:Sparse}
\textit{The underlying sparse matrix $\bm{S}\in \mathbb{R}^{m\times n}$ is $\alpha$-sparse. That is, $\bm{S}$ has at most $\alpha n$  non-zero entries in each row, and at most $\alpha m$ non-zero entries in each column.  In the other words, for all $1\leq i \leq m, 1\leq j \leq n$, 
\begin{equation}\label{eq:p_model}
\|\bm{e}_i^T \bm{S}\|_0 \leq \alpha n\quad\mbox{and} \quad \|\bm{S}\bm{e}_j\|_0 \leq \alpha m.
\end{equation}
In this paper, we assume\footnote{The standard notion ``$\lesssim$'' in \eqref{eq:condition_on_p} means there exists an absolute numerical constant $C>0$ such that $\alpha$ can be upper bounded by $C$ times the right hand side. }
\begin{equation}\label{eq:condition_on_p}
\alpha\lesssim\min \left\{ \frac{1}{\mu r^2 \kappa^3}, \frac{1}{\mu^{1.5} r^2\kappa},\frac{1}{\mu^2r^2}\right\},
\end{equation}
where $\kappa$ is the condition number of $\bm{L}$}.


Assumption \nameref{assume:Sparse} states that the non-zero entries of the sparse matrix $\BS$ cannot concentrate in a few rows or columns, so there does not exist a low rank component in $\BS$. If the indices of the support set $\Omega$ are sampled independently  from the Bernoulli distribution with the associated parameter being slightly smaller than $\alpha$, by the Chernoff inequality, one can easily show that  \eqref{eq:p_model} holds with high probability. 

\subsection{Organization and Notation of the Paper } \label{subsec:notation}
The rest of the paper is organized as follows. In the remainder of this section, we introduce standard notation  that is used throughout the paper. Section~\ref{subsec:proposed algorithms} presents the proposed algorithm and discusses how to implement it efficiently. The theoretical recovery guarantee of the proposed algorithm is presented in Section~\ref{subsec:guaranteed results}, followed by   a review of prior art for RPCA. In Section~\ref{sec:experience}, we present the numerical simulations of our algorithm.  
Section~\ref{sec:proofs} contains all the mathematical proofs of our main theoretical result. We conclude this paper with future directions in Section~\ref{sec:discussion}.

In this paper, vectors are denoted by bold lowercase letters (e.g., $\bm{x}$), matrices are denoted by bold capital letters (e.g., $\bm{X}$), and operators are denoted by calligraphic letters (e.g., $\mathcal{H}$). In particular, $\bm{e}_i$ denotes the $i^{th}$ canonical  basis vector, $\bm{I}$ denotes the identity matrix, and $\mathcal{I}$ denotes the identity operator. For a vector $\bm{x}$, $\|\bm{x}\|_0$ counts the number of non-zero entries in $\bm{x}$, and $\|\bm{x}\|_2$ denotes the $\ell_2$ norm of $\bm{x}$. For a matrix $\bm{X}$, $[\bm{X}]_{ij}$ denotes its $(i,j)^{th}$ entry, $\sigma_i(\bm{X})$ denotes its $i^{th}$ singular value, $\|\bm{X}\|_\infty=\max_{ij} |[\bm{X}]_{ij} |$ denotes the maximum magnitude of its entries, $\|\bm{X}\|_2=\sigma_1(\bm{X})$ denotes its spectral norm, $\|\bm{X}\|_F=\sqrt{\sum_i \sigma_i^2(\bm{X})}$ denotes its Frobenius norm, and $\|\bm{X}\|_*=\sum_i \sigma_i(\bm{X})$ denotes its nuclear norm. The inner product of two real valued vectors is defined as $\langle\bm{x},\bm{y}\rangle=\bm{x}^T\bm{y}$, and the inner product of two real valued matrices is defined as $\langle\bm{X},\bm{Y}\rangle=Trace(\bm{X}^T\bm{Y})$, where $(\cdot)^T$ represents the transpose of a vector or matrix.

Additionally, we sometimes use the shorthand $\sigma_i^A$ to denote the $i^{th}$ singular value of a matrix $\bm{A}$. Note that $\kappa=\sigma_{1}^L/\sigma_{r}^L$ always denotes the condition number of the underlying rank-$r$ matrix $\bm{L}$, and $\Omega=supp(\bm{S})$ is always referred to as the support of the underlying sparse matrix $\bm{S}$. At the $k^{th}$ iteration of the proposed algorithm, the estimates of the low rank matrix and the sparse matrix are denoted by $\bm{L}_k$ and $\bm{S}_k$, respectively.

\section{Algorithm and Theoretical Results} \label{sec:algo and results}
In this section, we present the new algorithm and its recovery guarantee. For ease of exposition, we assume all matrices are square (i.e., $m=n$), but emphasize that nothing is special about this assumption and all the results can be easily extended to rectangular matrices.
\subsection{Proposed Algorithm}  \label{subsec:proposed algorithms}
Alternating projections is a  minimization approach that has been successfully used in many fields, including image processing \citep{wang2008new,chan2000convergence,o2007alternating}, matrix completion \citep{keshavan2012efficient,jain2013low,hardt2013provable,tannerwei2016asd}, phase retrieval \citep{netrapalli2013phase,cai2017fast,zhang2017phase}, and many others \citep{peters2009interference,agarwal2014learning,yu2016alternating,pu2017complexity}. 
A  non-convex algorithm based on alternating projections, namely AltProj, is presented in \citep{netrapalli2014non} for RPCA accompanied with a theoretical recovery guarantee.  In each iteration, AltProj first updates $\bm{L}$ by projecting $\bm{D}-\bm{S}$ onto the space of rank-$r$ matrices, denoted  $\mathcal{M}_r$, and then updates $\bm{S}$ by projecting $\bm{D}-\bm{L}$ onto the space of sparse matrices, denoted $\mathcal{S}$; see  the left plot of Figure~\ref{fig:illustration} for an illustration. Regarding to the implementation of AltProj, 
the projection of a matrix onto the space of low rank matrices can be computed by the singular value decomposition (SVD) followed by truncating out small singular values, while the projection of a matrix onto the space of sparse matrices can be computed by the hard thresholding operator.
As a non-convex algorithm which targets \eqref{eq:non-convex model 2} directly, AltProj is computationally much more efficient than solving the convex relaxation problem \eqref{eq:convex model} using semidefinite programming (SDP). However, when projecting $\bm{D}-\bm{S}$ onto the low rank matrix manifold, AltProj requires to compute the SVD of a full size matrix, which is computationally expensive. Inspired by the work in  \citep{vandereycken2013low, wei2016guarantees_completion,wei2016guarantees_recovery}, we propose an accelerated algorithm for RPCA, coined accelerated alternating projections (AccAltProj), to circumvent the high computational cost of the SVD. The new algorithm is able to reduce the per-iteration computational cost of AltProj significantly, while  a theoretical guarantee can be similarly established.

\begin{figure}[t]
\subfloat[Illustration of AltProj\label{fig:AltProj}]
  {\includegraphics[width=.50\linewidth]{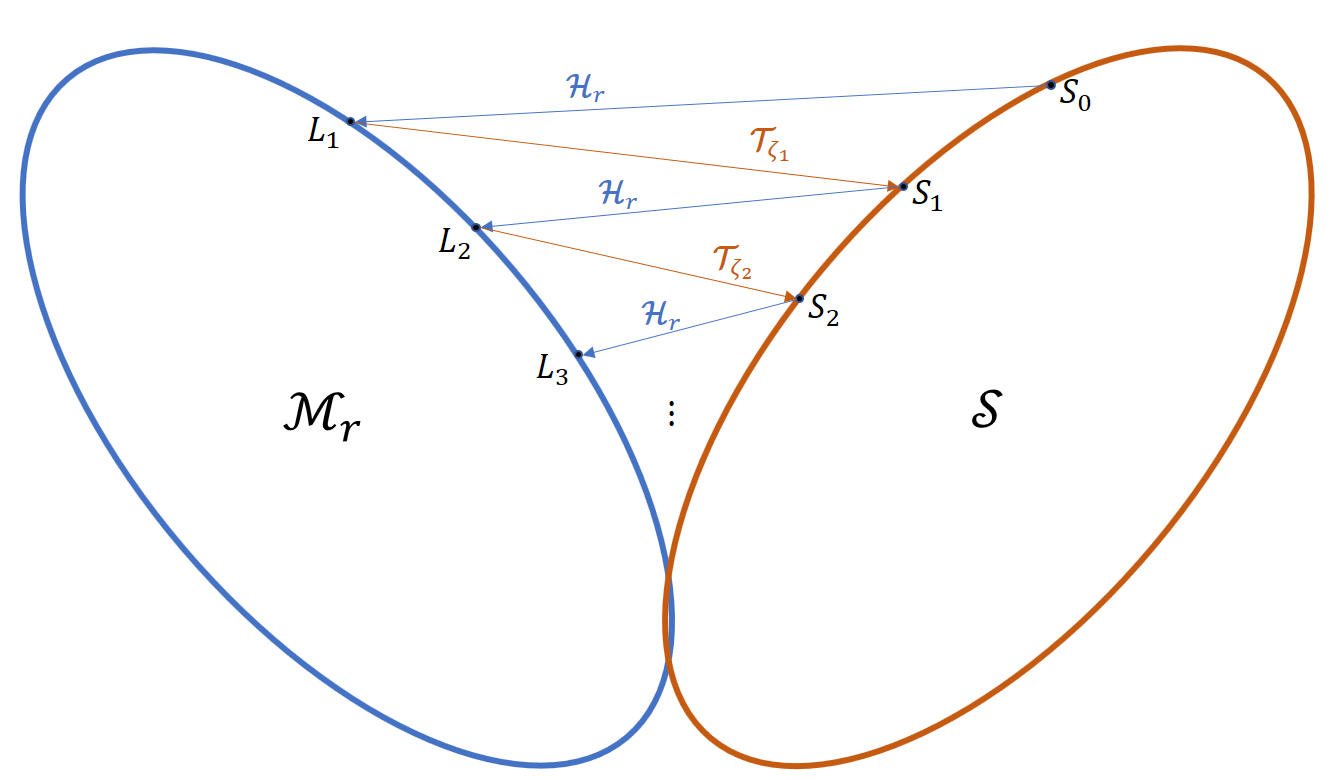}
  }\hfill
\subfloat[Illustration of AccAltProj\label{fig:AccAltProj}]
  {\includegraphics[width=.50\linewidth]{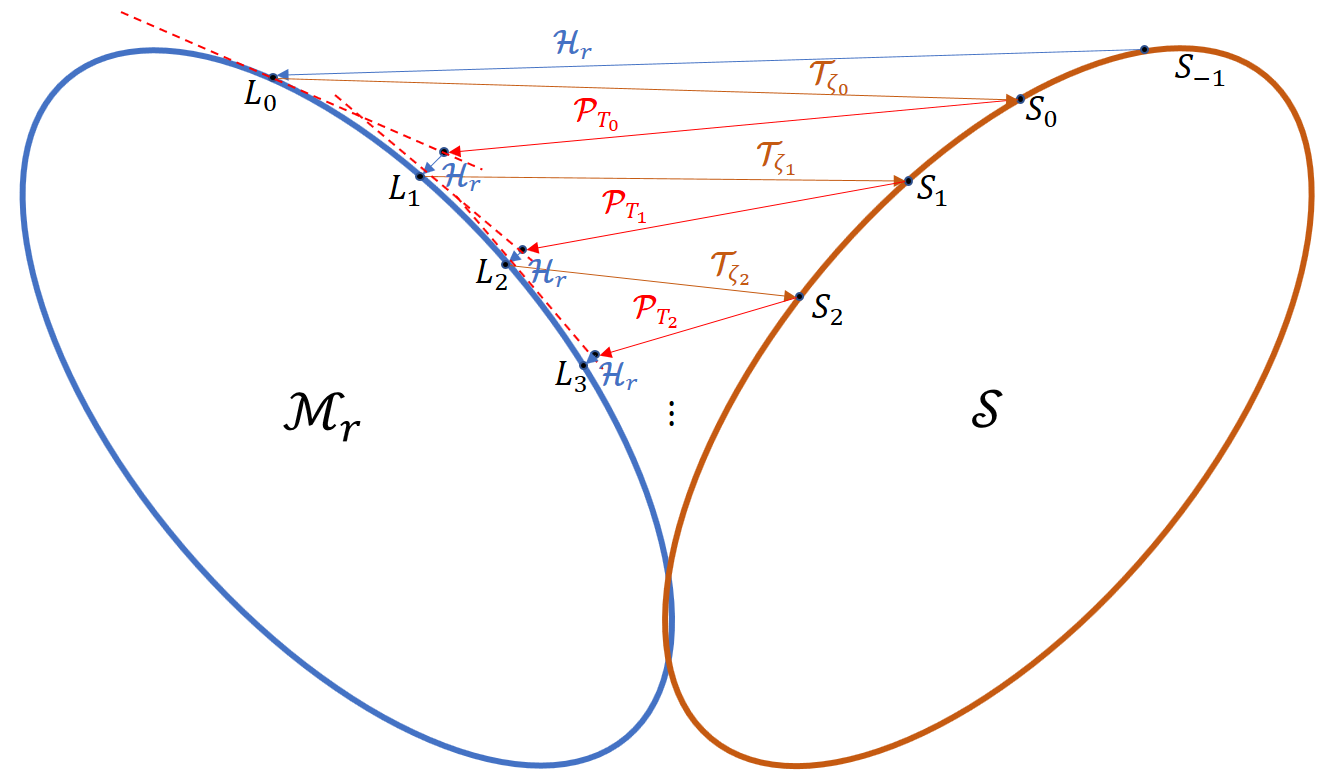}
  }\hfill
\caption{Visual comparison between  AltProj and AccAltProj, where $\mathcal{M}_r$ denotes the manifold of rank-$r$ matrices and $\mathcal{S}$ denotes the set of sparse matrices. The red dash line in \protect\subref{fig:AccAltProj} represents the tangent space of $\mathcal{M}_r$ at $\bm{L}_k$. In fact, each circle represents a sum of a low rank matrix and a sparse matrix, but with the component on one circle fixed when projecting onto the other circle. For conciseness, the trim stage, i.e., $\widetilde{\bm{L}}_k$, is not included in the plot for AccAltProj.}\label{fig:illustration}
\end{figure}

Our algorithm consists of two phases: initialization and projections onto $\mathcal{M}_r$ and $\mathcal{S}$ alternatively. We begin our discussion with the second phase, which is described in Algorithm~\ref{Algo:Algo1}. For geometric comparison between AltProj and AccAltProj, see Figure~\ref{fig:illustration}.

\begin{algorithm}[t]
\caption{Robust PCA by Accelerated Alternating Projections (AccAltProj)}\label{Algo:Algo1}
\begin{algorithmic}[1]
\State \textbf{Input:} $\bm{D}=\bm{L}+\bm{S}$: matrix to be split; $r$: rank of $\bm{L}$; $\epsilon$: target precision level; $\beta$: thresholding parameter; $\gamma$: target converge rate; $\mu$: incoherence parameter of $\bm{L}$.
\State \textbf{Initialization} 
\State $k=0$
\While{ \texttt{<$\|\bm{D}-\bm{L}_{k}-\bm{S}_{k}\|_F/\|\bm{D}\|_F \geq \epsilon$>} }
\State $\widetilde{\bm{L}}_{k}=\textnormal{Trim}(\bm{L}_{k},\mu)$
\State $\bm{L}_{k+1}=\mathcal{H}_r(\mathcal{P}_{\widetilde{T}_{k}}(\bm{D}-\bm{S}_{k}))$
\State $\zeta_{k+1}= \beta\left(\sigma_{r+1}\left(\mathcal{P}_{\widetilde{T}_{k}}(\bm{D}-\bm{S}_{k})\right) + \gamma^{k+1} \sigma_{1}\left(\mathcal{P}_{\widetilde{T}_{k}}(\bm{D}-\bm{S}_{k})\right)\right) $
\State $\bm{S}_{k+1}=\mathcal{T}_{\zeta_{k+1}}(\bm{D}-\bm{L}_{k+1})$
\State $k=k+1$
\EndWhile{\textbf{end while}}
\State \textbf{Output:} $\bm{L}_k$, $\bm{S}_k$
\end{algorithmic}
\end{algorithm}

\begin{algorithm}[t]
\caption{Trim}\label{Algo:Trim}
\begin{algorithmic}[1]
\State \textbf{Input:} $\bm{L}=\bm{U}\bm{\Sigma} \bm{V}^T$: matrix to be trimmed; $\mu$: target incoherence level.
\State $c_{\mu}=\sqrt{\frac{\mu r}{n}}$
\For{\texttt{<$i=1$ to $m$>}}
        \State $\bm{A}^{(i)}=\min\{1,\frac{c_{\mu}}{\|\bm{U}^{(i)}\|}\}\bm{U}^{(i)}$
\EndFor{\textbf{end for}}
\For{\texttt{<$j=1$ to $n$>}}
        \State $\bm{B}^{(j)}=\min\{1,\frac{c_{\mu}}{\|\bm{V}^{(j)}\|}\}\bm{V}^{(i)}$
\EndFor{\textbf{end for}}
\State \textbf{Output:} $\widetilde{\bm{L}}=\bm{A}\bm{\Sigma} \bm{B}$
\end{algorithmic}
\end{algorithm}

Let $(\BL_k,\BS_k)$ be a pair of current estimates. 
At the $(k+1)^{th}$ iteration, AccAltProj first trims $\BL_k$  into an incoherent matrix $\widetilde{\bm{L}}_k$ using Algorithm~\ref{Algo:Trim}. 
Noting that $\widetilde{\bm{L}}_k$ is still a rank-$r$ matrix, so its left and right singular vectors define an $(2n-r)r$-dimensional subspace \citep{vandereycken2013low}, 
\begin{equation}  \label{eq:tangent space tilde k}
\widetilde{T}_k=\{\widetilde{\bm{U}}_k\bm{A}^T+\bm{B}\widetilde{\bm{V}}_k^T ~|~\bm{A},\bm{B}\in\mathbb{R}^{n\times r} \},
\end{equation}
where $\widetilde{\bm{L}}_k=\widetilde{\bm{U}}_k\widetilde{\bm{\Sigma}}_k\widetilde{\bm{V}}_k^T$ is  the SVD of $\widetilde{\bm{L}}_k$\footnote{In practice, we only need the trimmed orthogonal matrices $\widetilde{\bm{U}}_k$ and $\widetilde{\bm{V}}_k$ for the projection $\mathcal{P}_{\widetilde{T}_k}$, and they can be computed efficiently via a QR decomposition. The entire matrix $\widetilde{\bm{L}}_k$ should never be formed in an efficient implementation of AccAltProj.}. Given a matrix $\bm{Z}\in\mathbb{R}^{n\times n}$, it can be easily verified that the projections of $\bm{Z}$ onto the subspace $\widetilde{T}_k$ and its orthogonal complement are given by
\begin{equation}  \label{eq:projection onto tangent space tilde k}
\mathcal{P}_{\widetilde{T}_k} \bm{Z}=\widetilde{\bm{U}}_k\widetilde{\bm{U}}_k^T\bm{Z}+\bm{Z}\widetilde{\bm{V}}_k\widetilde{\bm{V}}_k^T-\widetilde{\bm{U}}_k\widetilde{\bm{U}}_k^T\bm{Z}\widetilde{\bm{V}}_k\widetilde{\bm{V}}_k^T
\end{equation}
and
\begin{equation}  \label{eq:projection onto perpendicular space tilde k}
(\mathcal{I}-\mathcal{P}_{\widetilde{T}_k}) \bm{Z}=(\bm{I}-\widetilde{\bm{U}}_k\widetilde{\bm{U}}_k^T)\bm{Z}(\bm{I}-\widetilde{\bm{V}}_k\widetilde{\bm{V}}_k^T).
\end{equation}

As stated previously, AltProj truncates the SVD of $\bm{D}-\bm{S}_k$ directly to get a new estimate of $\BL$. { In contrast, AccAltProj
first projects $\bm{D}-\bm{S}_k$  onto the low dimensional subspace $\widetilde{T}_k$, and then projects the intermediate matrix onto the rank-$r$ matrix manifold $\mathcal{M}_r$ using the truncated SVD.} That is, 
$$\bm{L}_{k+1}=\mathcal{H}_r(\mathcal{P}_{\widetilde{T}_{k}}(\bm{D}-\bm{S}_{k})),$$
where $\mathcal{H}_r$ computes the best rank-$r$ approximation of a matrix, 
\begin{equation}
\mathcal{H}_r(\bm{Z}):=\bm{Q}\bm{\Lambda}_r\bm{P}^T \textnormal{ where }\bm{Z}=\bm{Q\Lambda P}^T\textnormal{ is its SVD and $[\bm{\Lambda}_r]_{ii}:=\begin{cases} [\bm{\Lambda}]_{ii} & i\leq r\\ 0& \mbox{otherwise}. \end{cases}$}
\end{equation}
Before proceeding, it is worth noting that the set of rank-$r$ matrices $\mathcal{M}_r$ form a smooth manifold of dimension $(2n-r)r$, and $\widetilde{T}_k$ is indeed the tangent space of $\mathcal{M}_r$ at $\widetilde{\bm{L}}_k$  \citep{vandereycken2013low}. Matrix manifold algorithms based on the tangent space of low dimensional spaces have been widely studied in the literature, see for example \citep{ngo2012scaled,mishra2012riemannian,vandereycken2013low,mishra2014r3mc,mishra2014fixed,wei2016guarantees_completion,wei2016guarantees_recovery} and references therein.  In particular, we invite  readers to explore the book  \citep{absil2009optimization} for more details about the  differential geometry ideas behind manifold algorithms. 

One can see that a SVD is still needed to obtain the new estimate $\bm{L}_{k+1}$. Nevertheless, it can be computed in a very efficient way \citep{vandereycken2013low,wei2016guarantees_completion,wei2016guarantees_recovery}. Let $(\bm{I}-\widetilde{\bm{U}}_k\widetilde{\bm{U}}_k^T)(\bm{D}-\bm{S}_k)\widetilde{\bm{V}}_k=\bm{Q}_1\bm{R}_1$ and $(\bm{I}-\widetilde{\bm{V}}_k\widetilde{\bm{V}}_k^T)(\bm{D}-\bm{S}_k)\widetilde{\bm{U}}_k=\bm{Q}_2\bm{R}_2$ be the QR decompositions of $(\bm{I}-\widetilde{\bm{U}}_k\widetilde{\bm{U}}_k^T)(\bm{D}-\bm{S}_k)\widetilde{\bm{V}}_k$ and $(\bm{I}-\widetilde{\bm{V}}_k\widetilde{\bm{V}}_k^T)(\bm{D}-\bm{S}_k)\widetilde{\bm{U}}_k$, respectively. Note that $(\bm{I}-\widetilde{\bm{U}}_k\widetilde{\bm{U}}_k^T)(\bm{D}-\bm{S}_k)\widetilde{\bm{V}}_k$ and $(\bm{I}-\widetilde{\bm{V}}_k\widetilde{\bm{V}}_k^T)(\bm{D}-\bm{S}_k)\widetilde{\bm{U}}_k$ can be computed by one matrix-matrix subtraction between an $n\times n$ matrix and an $n\times n$ matrix, two matrix-matrix multiplications between an $n\times n$ matrix and an $n\times r$ matrix, and a few matrix-matrix multiplications between a $r\times n$ and an $n\times r$ or  between an $n\times r$ matrix and a $r\times r$ matrix. Moreover,
A little algebra gives
\begin{align*}
\mathcal{P}_{\widetilde{T}_{k}} (\bm{D}-\bm{S}_k) &= \widetilde{\bm{U}}_k\widetilde{\bm{U}}_k^T(\bm{D}-\bm{S}_k)+(\bm{D}-\bm{S}_k)\widetilde{\bm{V}}_k\widetilde{\bm{V}}_k^T-\widetilde{\bm{U}}_k\widetilde{\bm{U}}_k^T(\bm{D}-\bm{S}_k)\widetilde{\bm{V}}_k\widetilde{\bm{V}}_k^T  \cr
			&=\widetilde{\bm{U}}_k\widetilde{\bm{U}}_k^T(\bm{D}-\bm{S}_k)(\bm{I}-\widetilde{\bm{V}}_k\widetilde{\bm{V}}_k^T)+(\bm{I}-\widetilde{\bm{U}}_k\widetilde{\bm{U}}_k^T)(\bm{D}-\bm{S}_k)\widetilde{\bm{V}}_k\widetilde{\bm{V}}_k^T+\widetilde{\bm{U}}_k\widetilde{\bm{U}}_k^T(\bm{D}-\bm{S}_k)\widetilde{\bm{V}}_k\widetilde{\bm{V}}_k^T  \cr
			&=\widetilde{\bm{U}}_k\bm{R}_2^T\bm{Q}_2^T+\bm{Q}_1\bm{R}_1\widetilde{\bm{V}}_k^T+\widetilde{\bm{U}}_k\widetilde{\bm{U}}_k^T(\bm{D}-\bm{S}_k)\widetilde{\bm{V}}_k\widetilde{\bm{V}}_k^T  \cr
			&=\begin{bmatrix}\widetilde{\bm{U}}_k & \bm{Q}_1\end{bmatrix} \begin{bmatrix}  \widetilde{\bm{U}}_k^T(\bm{D}-\bm{S}_k)\widetilde{\bm{V}}_k & \bm{R}_2^T  \\ \bm{R}_1 & \bm{0}   \end{bmatrix} \begin{bmatrix} \widetilde{\bm{V}}_k^T \\ \bm{Q}_2^T      \end{bmatrix}  \cr
			&:= \begin{bmatrix}\widetilde{\bm{U}}_k & \bm{Q}_1\end{bmatrix} \bm{M}_k \begin{bmatrix} \widetilde{\bm{V}}_k^T \\ \bm{Q}_2^T      \end{bmatrix},
\end{align*}
where the fourth line follows from the fact $\widetilde{\bm{U}}_k^T\bm{Q}_1=\widetilde{\bm{V}}_k^T\bm{Q}_2=\bm{0}$.
Let $\bm{M}_k = \bm{U}_{M_k}\bm{\Sigma}_{M_k}\bm{V}_{M_k}^T$ be the SVD of $\bm{M}_k$, which can be computed using $O(r^3)$ flops since $\bm{M}_k$ is a $2r\times 2r$ matrix. Then the  SVD of $\mathcal{P}_{\widetilde{T}_{k}} (\bm{D}-\bm{S}_k)=\widetilde{\bm{U}}_k\widetilde{\bm{\Sigma}}_k\widetilde{\bm{V}}_k^T$ can be computed by
\begin{equation}
\widetilde{\bm{U}}_{k+1}=\begin{bmatrix}\widetilde{\bm{U}}_k & \bm{Q}_1\end{bmatrix}\bm{U}_{M_k},\quad \widetilde{\bm{\Sigma}}_{k+1}=\bm{\Sigma}_{M_k},\quad\textnormal{and}\quad \widetilde{\bm{V}}_{k+1}=\begin{bmatrix}\widetilde{\bm{V}}_k & \bm{Q}_2\end{bmatrix}\bm{V}_{M_k}
\end{equation}
since both the matrices $\begin{bmatrix}\widetilde{\bm{U}}_k & \bm{Q}_1\end{bmatrix}$ and $\begin{bmatrix}\widetilde{\bm{V}}_k & \bm{Q}_2\end{bmatrix}$ are orthogonal. 
In summary,  the overall computational costs of $\mathcal{H}_r(\mathcal{P}_{\widetilde{T}_{k}}(\bm{D}-\bm{S}_{k}))$ lie in one matrix-matrix subtraction between an $n\times n$ matrix and an $n\times n$ matrix, two matrix-matrix multiplications between an $n\times n$ matrix and an $n\times r$ matrix, the QR decomposition of  two $n\times r$ matrices, an SVD of a $2r\times 2r$ matrix, and  a few matrix-matrix multiplications between a $r\times n$ matrix and an $n\times r$ matrix or between an $n\times r$ matrix and a $r\times r$ matrix, leading to a total of  $4n^2r+n^2+O(nr^2+r^3)$ flops. Thus, the dominant per iteration computational complexity of AccAltProj for updating the estimate of $\BL$ is the same as the novel gradient descent based approach introduced in \citep{yi2016fast}. In contrast, computing the best rank-$r$ approximation of a non-structured $n\times n$ matrix $\bm{D}-\bm{S}_k$ typically costs $O(n^2r)+n^2$ flops with a large hidden constant in front of $n^2r$.

After $\bm{L}_{k+1}$ is obtained, following the approach in \citep{netrapalli2014non}, we apply the hard thresholding operator to update the estimate of the sparse matrix, 
\begin{equation*}
\bm{S}_{k+1}=\mathcal{T}_{\zeta_{k+1}}(\bm{D}-\bm{L}_{k+1}),
\end{equation*}
where the thresholding operator $\mathcal{T}_{\zeta_{k+1}}$ is defined as
\begin{equation}
[\mathcal{T}_{\zeta_{k+1}}\bm{Z}]_{ij} =
\begin{cases}
[\bm{Z}]_{ij} & |[\bm{Z}]_{ij}| >\zeta_{k+1}\\
0  & \mbox{otherwise}
\end{cases}
\end{equation}
for any matrix $\bm{Z}\in\mathbb{R}^{m\times n}$. Notice that the thresholding value of $\zeta_{k+1}$ in Algorithm~\ref{Algo:Algo1} is chosen as
$$\zeta_{k+1}= \beta\left(\sigma_{r+1}\left(\mathcal{P}_{\widetilde{T}_{k}}(\bm{D}-\bm{S}_{k})\right) + \gamma^{k+1} \sigma_{1}\left(\mathcal{P}_{\widetilde{T}_{k}}(\bm{D}-\bm{S}_{k})\right)\right) ,$$
which relies on a tuning parameter $\beta>0$, a convergence rate parameter $0\leq\gamma<1$, and the singular values of $\mathcal{P}_{\widetilde{T}_{k}} (\bm{D}-\bm{S}_k)$. Since we have already obtained all the singular values of $\mathcal{P}_{\widetilde{T}_{k}} (\bm{D}-\bm{S}_k)$ when computing $\bm{L}_{k+1}$,  the extra cost of computing $\zeta_{k+1}$ is very marginal. Therefore, the cost of updating the estimate of $\BS$ is very low and insensitive to the sparsity of $\bm{S}$.

In this paper, a good initialization is achieved by two steps of  modified AltProj  when setting the input rank to $r$, see Algorithm~\ref{Algo:Init1}. With this initialization scheme, we can construct an initial guess that is sufficiently close to the ground truth and is inside the ``basin of attraction'' as detailed in the next subsection. Note that the thresholding parameter $\beta_{init}$ used in  Algorithm~\ref{Algo:Init1} is different from that in Algorithm~\ref{Algo:Algo1}.

\begin{algorithm}[htp]
\caption{Initialization by Two Steps of  AltProj}\label{Algo:Init1}
\begin{algorithmic}[1]
\State \textbf{Input:} $\bm{D}=\bm{L}+\bm{S}$: matrix to be split; $r$: rank of $\bm{L}$; $\beta_{init}, \beta$: thresholding parameters.
\State $\bm{L}_{-1}=\bm{0}$
\State $\zeta_{-1} = \beta_{init} \cdot \sigma_1^D$
\State $\bm{S}_{-1}=\mathcal{T}_{\zeta_{-1}}(\bm{D}-\bm{L}_{-1})$
\State $\bm{L}_0=\mathcal{H}_r(\bm{D}-\bm{S}_{-1})$
\State $\zeta_{0} = \beta \cdot \sigma_{1}(\bm{D}-\bm{S}_{-1})$
\State $\bm{S}_0=\mathcal{T}_{\zeta_{0}}(\bm{D}-\bm{L}_{0})$
\State \textbf{Output:} $\bm{L}_0$, $\bm{S}_0$
\end{algorithmic}
\end{algorithm}




\subsection{Theoretical Guarantee} \label{subsec:guaranteed results}




In this subsection, we present the theoretical recovery guarantee of AccAltProj (Algorithm~\ref{Algo:Algo1} together with Algorithm~\ref{Algo:Init1}).  
The following theorem establishes the 
local convergence of AccAltProj.

\begin{theorem}[Local Convergence of AccAltProj]  \label{thm:local convergence}
 Let $\bm{L}\in\mathbb{R}^{n\times n}$ and $\bm{S}\in\mathbb{R}^{n\times n}$ be two symmetric matrices satisfying Assumptions \nameref{assume:Inco} and \nameref{assume:Sparse}. If the  initial guesses $\bm{L}_0$ and $\bm{S}_0$ obey the following conditions:
\[
\|\bm{L}-\bm{L}_0\|_2 \leq 8\alpha\mu r \sigma_1^L,\quad
\|\bm{S}-\bm{S}_0\|_\infty \leq \frac{\mu r}{n} \sigma_1^L,\quad \textnormal{and} \quad
supp(\bm{S}_0)\subset \Omega,
\]
then 
the iterates of Algorithm \ref{Algo:Algo1} with parameters  $\beta =\frac{\mu r}{2n}$ and $\gamma\in\left(\frac{1}{\sqrt{12}},1\right)$ satisfy 
\[
\|\bm{L}-\bm{L}_k\|_2 \leq 8\alpha\mu r \gamma^k\sigma_1^L,\quad
\|\bm{S}-\bm{S}_k\|_\infty \leq \frac{\mu r}{n} \gamma^k\sigma_1^L,\quad \textnormal{and} \quad
supp(\bm{S}_k)\subset \Omega.
\]
\end{theorem}



The next theorem states that the initial guesses obtained from Algorithm \ref{Algo:Init1} fulfill the conditions required in Theorem~\ref{thm:local convergence}.

\begin{theorem}[Guaranteed Initialization]\label{thm:initialization bound} Let $\bm{L}\in\mathbb{R}^{n\times n}$ and $\bm{S}\in\mathbb{R}^{n\times n}$ be two symmetric matrices satisfying Assumptions \nameref{assume:Inco} and \nameref{assume:Sparse}, respectively. If the thresholding parameters obey $\frac{\mu r\sigma_1^L}{n\sigma_1^D}\leq\beta_{init}\leq\frac{3\mu r\sigma_1^L}{n\sigma_1^D}$ and $\beta=\frac{{\mu r}}{2n}$, then the outputs of  Algorithm \ref{Algo:Init1} satisfy
\[
\|\bm{L}-\bm{L}_0\|_2 \leq 8\alpha\mu r \sigma_1^L,\quad
\|\bm{S}-\bm{S}_0\|_\infty \leq \frac{\mu r}{n} \sigma_1^L,\quad \textnormal{and} \quad
supp(\bm{S}_0)\subset \Omega.
\]
\end{theorem}

The proofs of Theorems \ref{thm:local convergence} and \ref{thm:initialization bound} are presented in Section~\ref{sec:proofs}. The convergence of AccAltProj follows immediately by combining the above  two theorems together.

For conciseness, the main theorems are stated for symmetric matrices. 
However, similar results can be established for nonsymmetric matrix recovery problems as they can be cast as  problems with respect to  symmetric  augmented matrices, as suggested in \citep{netrapalli2014non}.
Without loss of generality, assume $dm\leq n < (d+1)m$ for some $d\geq 1$ and construct $\overline{\BL}$ and $\overline{\BS}$ as
\begin{equation*} \setstretch{1.5}
\overline{\BL}:=
\begin{bmatrix}\,
\smash{
\underbrace{
\begin{matrix}
\bm{0} &\cdots & \bm{0} \\
\vdots &\ddots &\vdots  \\
\bm{0} &\cdots &\bm{0}  \\
\BL^T  &\cdots &\BL^T   
\end{matrix}}_{d \textnormal{ times}}  }
\begin{matrix}
~&\BL\\
~&\vdots\\
~&\BL\\
~&\bm{0}
\end{matrix}
\vphantom{
  \begin{matrix}
  \smash[b]{\vphantom{\Big|}}
  0\\\vdots\\0\\0
  \smash[t]{\vphantom{\Big|}}
  \end{matrix}
}
\,\,\end{bmatrix}
\begin{matrix}
\left.
\vphantom{\begin{matrix}\BL\\\vdots\\\BL\end{matrix}}
\right\rbrace{\scriptstyle {d \textnormal{ times}}}\\~\\
\end{matrix}, \qquad
\overline{\BS}:=
\begin{bmatrix}\,
\smash{
\underbrace{
\begin{matrix}
\bm{0} &\cdots & \bm{0} \\
\vdots &\ddots &\vdots  \\
\bm{0} &\cdots &\bm{0}  \\
\BS^T  &\cdots &\BS^T   
\end{matrix}}_{d \textnormal{ times}}  }
\begin{matrix}
~&\BS\\
~&\vdots\\
~&\BS\\
~&\bm{0}
\end{matrix}
\vphantom{
  \begin{matrix}
  \smash[b]{\vphantom{\Big|}}
  0\\\vdots\\0\\0
  \smash[t]{\vphantom{\Big|}}
  \end{matrix}
}
\,\,\end{bmatrix}
\begin{matrix}
\left.
\vphantom{\begin{matrix}\BS\\\vdots\\\BS\end{matrix}}
\right\rbrace{\scriptstyle {d \textnormal{ times}}}\\~\\
\end{matrix}.
\end{equation*}
\\\\
Then it is not hard to see that $\overline{\BL}$ is $O(\mu)$-incoherent, and $\overline{\BS}$ is $O(\alpha)$-sparse, with the hidden constants being independent of $d$. Moreover, based on the connection between the SVD of the augmented matrix and that of the original one, it can be easily verified that at the $k^{th}$ iteration  the estimates returned by AccAltProj with input $\overline{\BD}=\overline{\BL}+\overline{\BS}$ have the form \begin{equation*} \setstretch{1.5}
\overline{\BL}_k=
\begin{bmatrix}\,
\smash{
\underbrace{
\begin{matrix}
\bm{0} &\cdots & \bm{0} \\
\vdots &\ddots &\vdots  \\
\bm{0} &\cdots &\bm{0}  \\
\BL_k^T  &\cdots &\BL_k^T   
\end{matrix}}_{d \textnormal{ times}}  }
\begin{matrix}
~&\BL_k\\
~&\vdots\\
~&\BL_k\\
~&\bm{0}
\end{matrix}
\vphantom{
  \begin{matrix}
  \smash[b]{\vphantom{\Big|}}
  0\\\vdots\\0\\0
  \smash[t]{\vphantom{\Big|}}
  \end{matrix}
}
\,\,\end{bmatrix}
\begin{matrix}
\left.
\vphantom{\begin{matrix}\BL_k\\\vdots\\\BL_k\end{matrix}}
\right\rbrace{\scriptstyle {d \textnormal{ times}}}\\~\\
\end{matrix}, \qquad
\overline{\BS}_k=
\begin{bmatrix}\,
\smash{
\underbrace{
\begin{matrix}
\bm{0} &\cdots & \bm{0} \\
\vdots &\ddots &\vdots  \\
\bm{0} &\cdots &\bm{0}  \\
\BS_k^T  &\cdots &\BS_k^T   
\end{matrix}}_{d \textnormal{ times}}  }
\begin{matrix}
~&\BS_k\\
~&\vdots\\
~&\BS_k\\
~&\bm{0}
\end{matrix}
\vphantom{
  \begin{matrix}
  \smash[b]{\vphantom{\Big|}}
  0\\\vdots\\0\\0
  \smash[t]{\vphantom{\Big|}}
  \end{matrix}
}
\,\,\end{bmatrix}
\begin{matrix}
\left.
\vphantom{\begin{matrix}\BS_k\\\vdots\\\BS_k\end{matrix}}
\right\rbrace{\scriptstyle {d \textnormal{ times}}}\\~\\
\end{matrix},
\end{equation*}
\\\\
where $\BL_k,\BS_k$ are the the $k^{th}$ estimates returned by AccAltProj with input $\BD=\bm{L}+\bm{S}$.


\subsection{Related Work}  \label{subsec:related work}
As mentioned earlier, convex relaxation based  methods for RPCA have higher computational complexity and slower convergence rate which are not applicable for high dimensional problems. 
In fact, the convergence rate of the algorithm for computing the solution to the SDP formulation of RPCA  \citep{candes2011robust,chandrasekaran2011rank,xu2010robust} is sub-linear with the per iteration 
computational complexity being $O(n^3)$. By contrast, AccAltProj only requires $O(\log(1/\epsilon))$ iterations to achieve an accuracy of $\epsilon$, and the dominant per iteration computational cost is $O(rn^2)$.

There have been many other algorithms which are designed to solve the non-convex RPCA problem directly. 
In \citep{WSLer2013}, an alternating minimization algorithm was proposed for  \eqref{eq:non-convex model 2} based on the factorization model of 
low rank matrices. However, only convergence to fixed points was established there. 
In \citep{gu2016low}, the authors developed an alternating minimization algorithm for RPCA, which allows the sparsity level $\alpha$ to be $O(1/(\mu^{2/3}r^{2/3}n))$  for successful recovery, which is more stringent than our result when $r\ll n$. A projected gradient descent algorithm was proposed in  \cite{chen2015fast} for the special case of positive semidefinite matrices based on the $\ell_1$-norm of each row of the underlying sparse matrix, which is not very practical.  
 

In Table~\ref{tab:algo compare}, we compare AccAltProj with the other two competitive non-convex algorithms for RPCA: AltProj from \citep{netrapalli2014non} and non-convex gradient descent (GD) from \citep{yi2016fast}. GD attempts to reconstruct the low rank matrix by minimizing an objective function which contains the prior knowledge of the sparse matrix. The table displays the computational complexity of each algorithm for updating the estimates of the low rank matrix and the sparse matrix, as well as the convergence rate  and the theoretical tolerance for the number of non-zero entries in the sparse matrix.

From the table, we can see that AccAltProj achieves the same linear convergence rate as AltProj, which is faster than GD. Moreover, AccAltProj has the lowest per iteration computational complexity for updating both the estimates of $\BL$ and $\BS$ (ties with AltProj for updating the sparse part).
It is worth emphasizing that the acceleration stage in AccAltProj which first projects $\bm{D}-\bm{S}_k$ onto a low dimensional subspace reduces the computational cost of the SVD in AltProj dramatically. 
Overall, 
AccAltProj will be substantially faster than AltProj and GD, as confirmed by our numerical simulations in next section. The table also shows that the theoretical 
sparsity level that can be tolerated by AccAltProj is lower than that of GD and AltProj.  Our result looses an order in $r$ because we have
replaced the spectral norm by the Frobenius norm when considering the reduction of the reconstruction error in terms of the spectral norm. In addition, the condition number of the target matrix appears in the theoretical result because the current version
of AccAltProj deals with the fixed rank case which requires the initial guess is sufficiently
close to the target matrix for the theoretical analysis.
Nevertheless, we note that the sufficient condition regarding to $\alpha$ to guarantee the exact recovery of AccAltProj is highly pessimistic when compared with its empirical performance. Numerical investigations in next section show that AccAltProj can tolerate as large $\alpha$ as AltProj does under different energy  levels.

\begin{table}[htp]
\centering
\caption{Comparison of AccAltProj, AltProj and GD.} \label{tab:algo compare}
\makegapedcells
\setcellgapes{3pt}
\begin{tabular}{ |c||c|c|c|c| }
\hline
Algorithm  &   AccAltProj      &    AltProj      & GD \cr

\hhline {|=||=|=|=|=|}
Updating $\bm{S}$ & $\bm{O\left(n^2\right)}$        & $O\left(rn^2\right)$& $O\left(n^2+\alpha n^2\log(\alpha n)\right)$ \cr
Updating $\bm{L}$   & $\bm{O\left(rn^2\right)}$        & $O\left(r^2n^2\right)$ & $\bm{O\left(rn^2\right)}$                                               \cr
Tolerance of $\alpha$                                    & $O\left(\frac{1}{\max\{\mu r^2 \kappa^3,\mu^{1.5} r^2\kappa,\mu^2r^2\}}\right)$& $\bm{O\left(\frac{1}{\mu r}\right)}$& $O\left(\frac{1}{\max\{\mu r^{1.5}\kappa^{1.5},\mu r\kappa^2\}}\right)$        \cr
Iterations needed                                   &
$\bm{O\left(\log(\frac{1}{\epsilon})\right)}$  &  $\bm{O\left(\log(\frac{1}{\epsilon})\right)}$   &
$O\left(\kappa\log(\frac{1}{\epsilon})\right)$  \cr
\hline
\end{tabular}
\end{table}

\section{Numerical Expierments}  \label{sec:experience}
In this section, we present the empirical performance of our AccAltProj algorithm and compare it with the state-of-the-art AltProj algorithm from \citep{netrapalli2014non} and the leading 
gradient descent based algorithm (GD) from \citep{yi2016fast}. 
The tests are conducted  on a laptop equipped with 64-bit Windows 7, Intel i7-4712HQ (4 Cores at 2.3 GHz) and 16GB DDR3L-1600 RAM, and executed from MATLAB R2017a. We implement AltProj  by ourselves, while the codes for GD are downloaded from the author's website\footnote{Website: \url{www.yixinyang.org/code/RPCA_GD.zip}.}. Hand tuned parameters 
are used for these algorithms to achieve the best performance in the numerical comparison. The codes for AccAltProj can be found online:
\begin{center}
\url{https://github.com/caesarcai/AccAltProj_for_RPCA}.
\end{center}

Notice that the computation of an initial guess by Algorithm~\ref{Algo:Init1} requires the truncated SVD on a full size matrix. As is typical in the literature, we used the PROPACK library\footnote{Website: \url{sun.stanford.edu/~rmunk/PROPACK}.} for this task  when the size of $\bm{D}$ is large and $r$ is relatively small.  
To reduce the dependence of the theoretical result on the condition number of the underlying low rank matrix,  AltProj was originally designed to loop $r$ stages  for the input rank increasing from $1$ to $r$ and each stage contains a few number of iterations for a fixed rank. 
However, when the condition number is medium large which is the case in our tests, we have observed that AltProj achieves the best computational  efficiency when fixing the rank to $r$. Thus, to make fair comparison, we test AltProj when input rank is fixed, the same as the other two algorithms. 

\paragraph*{Synthetic Datasets} We follow the setup in \citep{netrapalli2014non} and \citep{yi2016fast} for the random tests on synthetic data.  
An $n\times n$ rank $r$ matrix $\bm{L}$ is formed via $\bm{L}=\bm{P}\bm{Q}^T$, where  $\bm{P},\bm{Q}\in\mathbb{R}^{n\times r}$ are two random matrices having their entries drawn i.i.d from the standard
normal distribution.
 The locations of the non-zero entries of the underlying sparse matrix $\bm{S}$ are sampled uniformly and independently without replacement, while the values of the non-zero entries are drawn i.i.d from the uniform distribution over the interval $[-c\cdot\mathbb{E}(|[\bm{L}]_{ij}|),c\cdot\mathbb{E}(|[\bm{L}]_{ij}|)]$ for some constant $c>0$. 
The relative computing error at the $k^{th}$ iteration of a single test is defined as
\begin{equation} \label{relative_err}
err_k=\frac{\|\bm{D}-\bm{L}_k-\bm{S}_k\|_F}{\|\bm{D}\|_F} .
\end{equation}
The test algorithms are terminated when either the relative computing error is smaller than a tolerance, $err_k<tol$, or a maximum number of $100$ iterations is reached. Recall that $\mu$ is the incoherence parameter of the low rank matrix $\BL$ and $\alpha$ is the sparsity parameter of the sparse matrix $\BS$. In the random tests, we use $1.1\mu$ in AltProj and AccAltProj, and use $1.1\alpha$ in GD.

Though we are only able to provide a theoretical guarantee for AccAltProj with trim in this paper, it can be easily seen that AccAltProj can also be implemented without the trim step. Thus, both AccAltProj with and without trim are tested. The parameters $\beta$ and $\beta_{init}$ are set to be
$\beta=\frac{1.1\mu r}{2\sqrt{mn}}$ and $\beta_{init}=\frac{1.1\mu r}{\sqrt{mn}}$ in our experiments, and $\gamma=0.5$ is used when $\alpha<0.55$ and $\gamma=0.65$ is used when $\alpha\geq 0.55$.

We first test the performance of the algorithms under different values of $\alpha$ for fixed $n=2500$ and $r=5$. Three different values of $c$ are investigated: $c\in\{0.2,1,5\}$, which represent three  different signal levels of $\BS$. For each value of $c$, $10$ different values of $\alpha$ from $0.3$ to $0.75$ are tested. We set $tol=10^{-6}$ in the stopping condition for all the test algorithms. The backtracking line search has been used in GD which can improve its recovery performance substantially in our tests. An algorithm is considered to have successfully reconstructed $\BL$ (or equivalently, $\BS$) if the low rank output of the algorithm $\BL_k$ satisfies 
\begin{align*}
\frac{\|\BL_k-\BL\|_F}{\|\BL\|_F}\leq 10^{-4}.
\end{align*}
The number of successful reconstructions for each algorithm out of $10$ random tests are presented in Table~\ref{tab:phase compare}. It is clear that AccAltProj (with and without trim) and AltProj exhibit similar behavior even though the theoretical requirement of AccAltProj with trim is a bit more stringent than that of AltProj, and they can tolerate larger values of $\alpha$ than GD when $c$ is small.

Next, we evaluate the runtime of the test algorithms. The computational results are plotted in Figure~\ref{synthetic_test_result} together with the setup corresponding to each plot.  
Figure~\ref{synthetic_test_result}\protect\subref{fig:varying_n2} shows that AccAltProj is substantially faster than AltProj and GD. In particular, when $n$ is large, it achieves about $10\times$ speedup. Figure~\ref{synthetic_test_result}\protect\subref{fig:varying_p2} shows that AccAltProj and AltProj are less sensitive to the sparsity  of $\bm{S}$. Notice that we have used a well-tuned fixed stepsize for GD here so that it can  achieve the best computational efficiency. Thus, GD fails to converge when $\alpha\geq 0.35$ which is smaller than the largest value of $\alpha$  for successful recovery corresponding to $c=1$   in Table~\ref{tab:phase compare}. Lastly, Figure~\ref{synthetic_test_result}\protect\subref{fig:conv_speed2} shows the lowest computational time of AccAltProj against the relative computing error. 
\begin{table}

\caption{Rate of success for AccAltProj with and without trim, AltProj, and GD for different values of $\alpha$.} \label{tab:phase compare} 

\centering

\begin{tabular}{ |c||c|c|c|c|c|c|c|c|c|c| } 

\hline

$c=0.2$             & 0.3 & 0.35 & 0.4 & 0.45& 0.5 & 0.55 & 0.6 & 0.65 & 0.7 & 0.75\cr

\hhline {|=||=|=|=|=|=|=|=|=|=|=|}

AccAltProj w/ trim  & 10  & 10   & 10  & 10  & 10  & 10   & 10  & 10   & 4   & 0   \cr
 
AccAltProj w/o trim & 10  & 10   & 10  & 10  & 10  & 10   & 10  & 10   & 4   & 0   \cr

AltProj             & 10  & 10   & 10  & 10  & 10  & 10   & 10  & 10   & 0   & 0   \cr

GD                  & 10  & 10   & 10  &  0  &  0  &  0   &  0  & 0    & 0   & 0   \cr

\hline

\end{tabular}

\vspace{5pt}

\begin{tabular}{ |c||c|c|c|c|c|c|c|c|c|c| } 

\hline

$c=1$               & 0.3 & 0.35 & 0.4 & 0.45& 0.5 & 0.55 & 0.6  & 0.65 & 0.7 & 0.75 \cr

\hhline {|=||=|=|=|=|=|=|=|=|=|=|}

AccAltProj w/ trim  & 10  & 10   & 10  & 10  & 10  & 10   & 10   & 9    & 0   & 0    \cr

AccAltProj w/o trim & 10  & 10   & 10  & 10  & 10  & 10   & 10   & 9    & 0   & 0    \cr

AltProj             & 10  & 10   & 10  & 10  & 10  & 10   & 10   & 8    & 0   & 0    \cr

GD                  & 10  & 10   & 10  & 10  &  9  &  0   &  0   & 0    & 0   & 0    \cr

\hline

\end{tabular}

\vspace{5pt}

\begin{tabular}{ |c||c|c|c|c|c|c|c|c|c|c| } 

\hline

$c=5$               & 0.3 & 0.35 & 0.4 & 0.45& 0.5 & 0.55 & 0.6 & 0.65  & 0.7 & 0.75\cr

\hhline {|=||=|=|=|=|=|=|=|=|=|=|}

AccAltProj w/ trim  & 10  & 10   & 10  & 10  & 10  & 10   & 10  & 5     & 0   & 0   \cr

AccAltProj w/o trim & 10  & 10   & 10  & 10  & 10  & 10   & 10  & 5     & 0   & 0   \cr

AltProj             & 10  & 10   & 10  & 10  & 10  & 10   & 10  & 3     & 0   & 0   \cr

GD                  & 10  & 10   & 10  & 10  & 10  & 10   &  7  & 0     & 0   & 0   \cr

\hline

\end{tabular}
\end{table}

\begin{figure}
\subfloat[Varying Dimension vs Runtime\label{fig:varying_n2}]
  {\includegraphics[width=.32\linewidth]{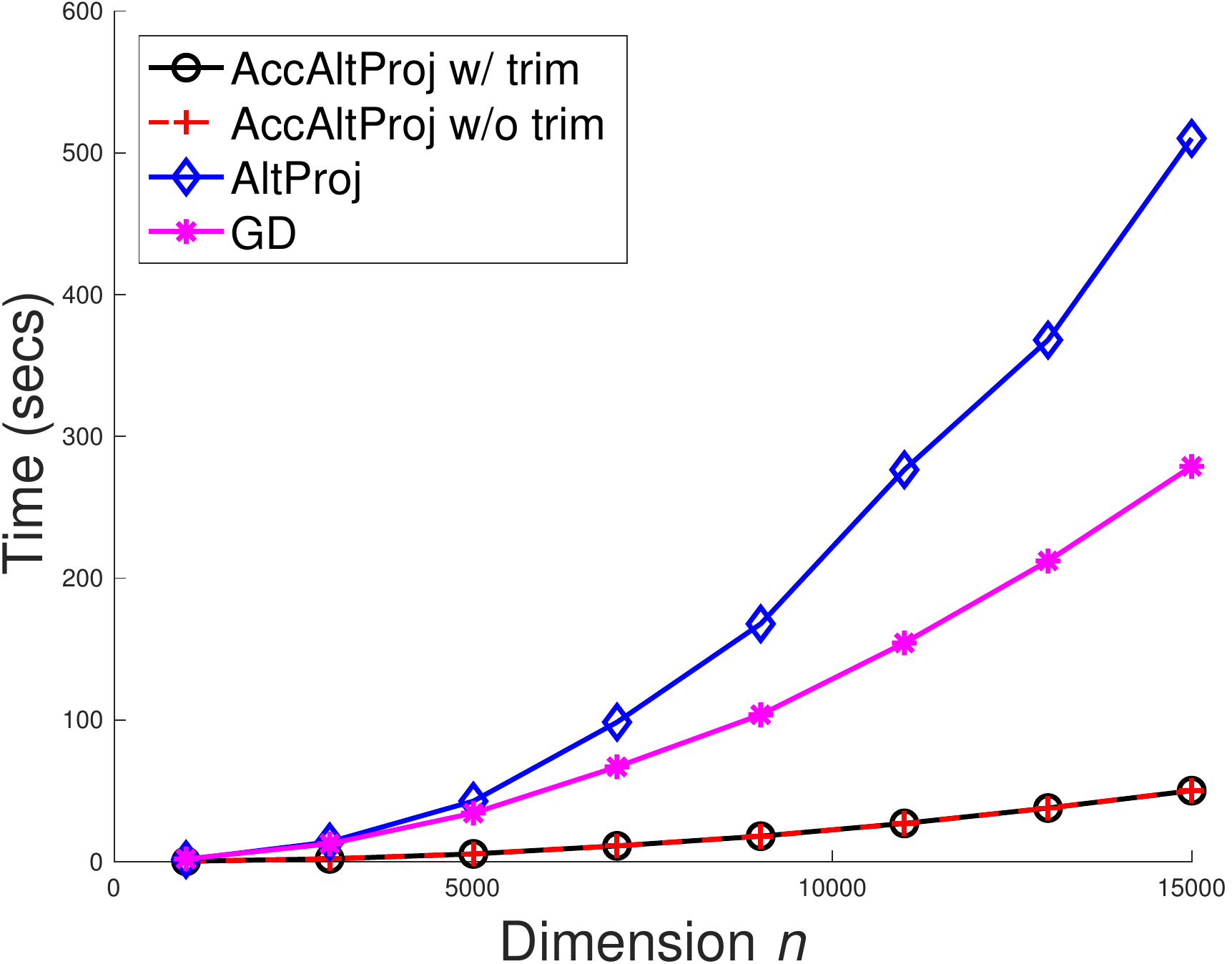}
  }\hfill
\subfloat[Varying Sparsity vs Runtime\label{fig:varying_p2}]
  {\includegraphics[width=.32\linewidth]{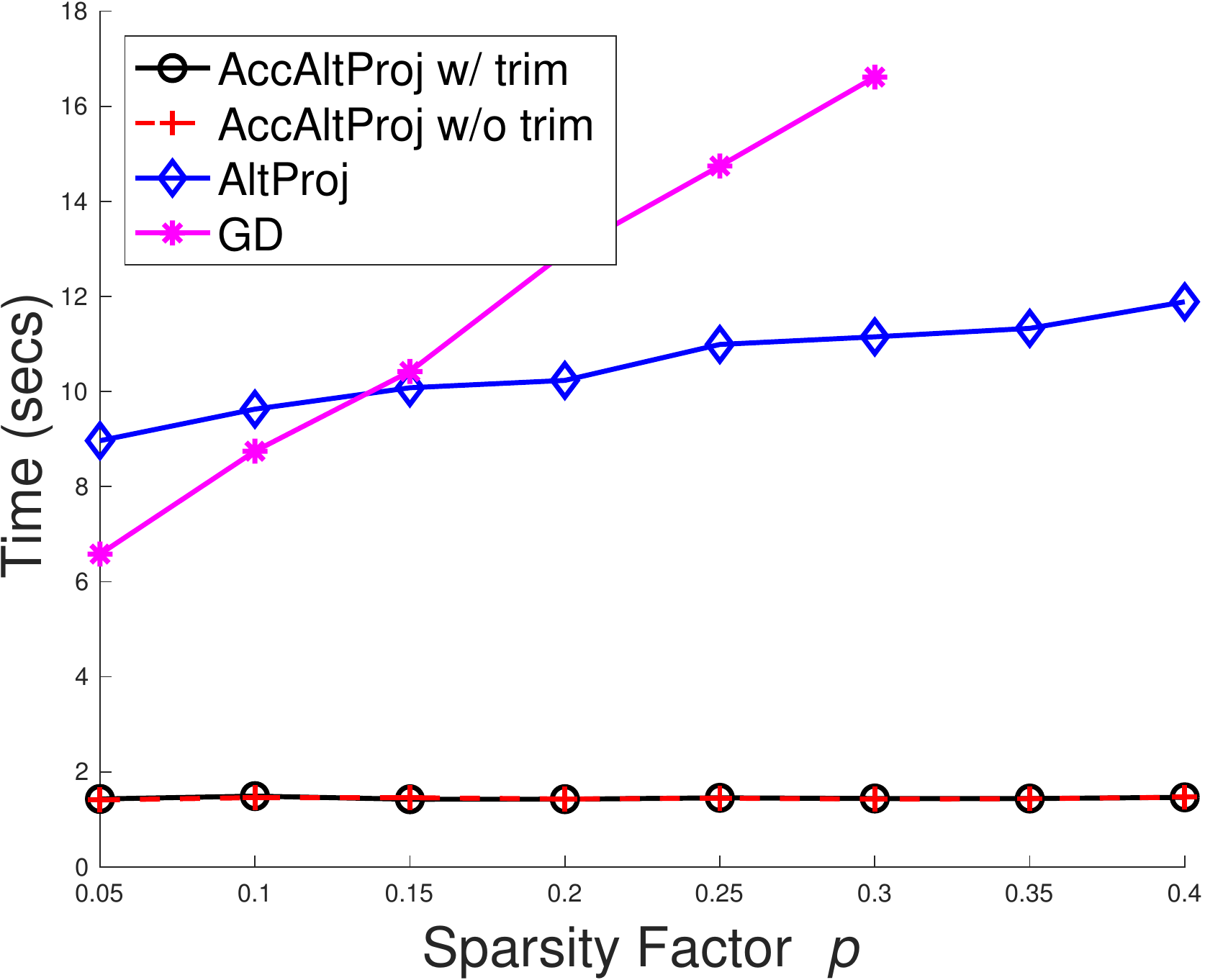}
  }\hfill
\subfloat[Relative Error vs Runtime\label{fig:conv_speed2}]
  {\includegraphics[width=.32\linewidth]{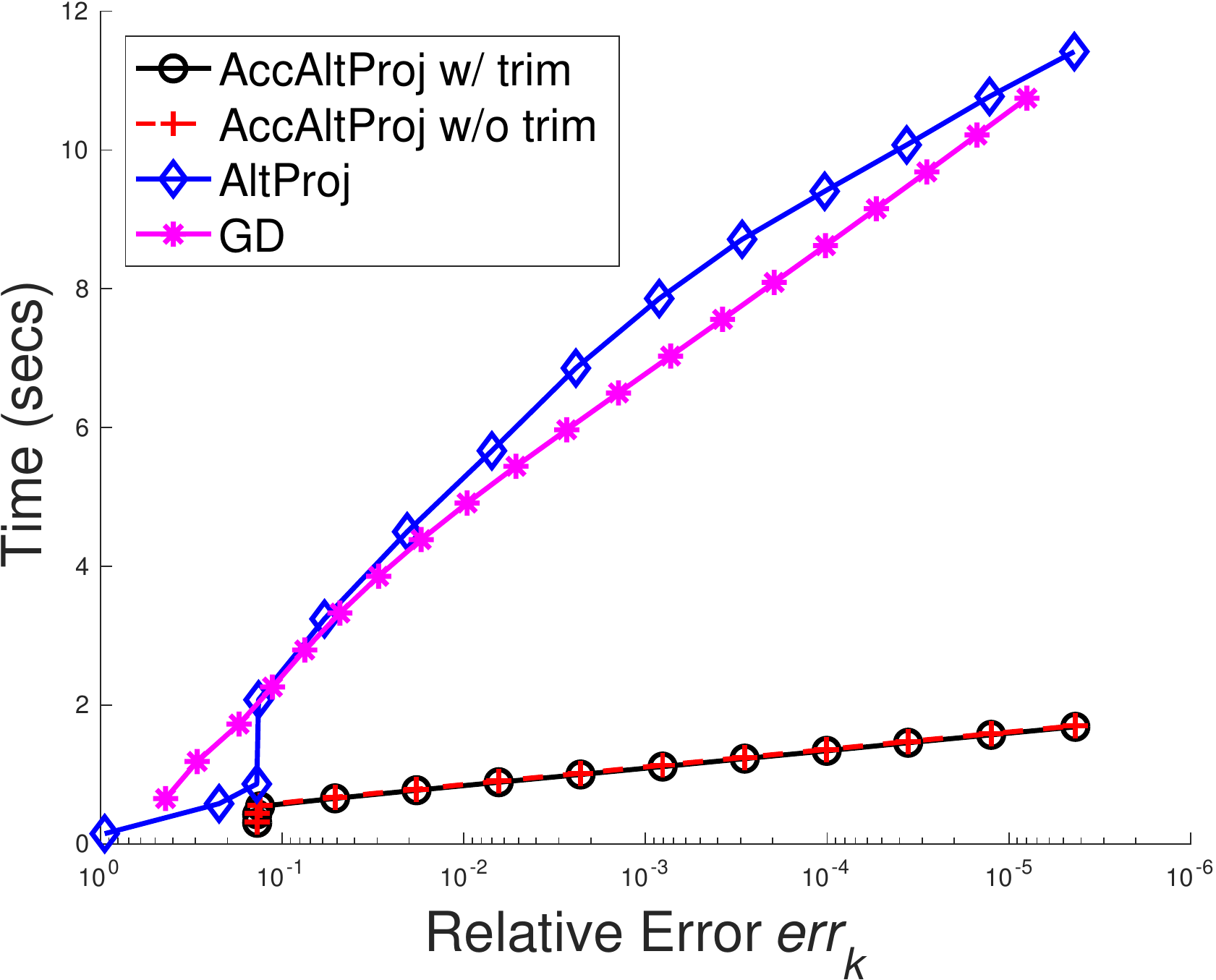}}
\caption{Runtime for synthetic datasets:
\protect\subref{fig:varying_n2} Varying dimension $n$ vs runtime, where $r=5$, $\alpha=0.1$, $c=1$, and $n$ varies from $1000$ to $15000$. The algorithms are terminated after $err_k<10^{-4}$ is satisfied.
\protect\subref{fig:varying_p2} Varying sparsity factor $\alpha$ vs runtime, where $r=5$, $c=1$ and $n=2500$. The  algorithms are terminated when either $err_k<10^{-4}$ or 100 number of iterations is reached, whichever comes first. \protect\subref{fig:conv_speed2} Relative error $err_k$ vs runtime, where  $r=5$, $\alpha=0.1$,  $c=1$, and $n=2500$. The algorithms are terminated after $err_k<10^{-5}$ is satisfied so that we can observe more iterations.}\label{synthetic_test_result}
\end{figure}

\paragraph*{Video Background Subtraction} 
In this section, we compare the performance of AccAltProj with and without trim, AltProj and GD on video background subtraction, a real world benchmark problem for RPCA. The task in background subtraction is  to separate  moving foreground objects from a static background. The two videos we have used for this test are \textit{Shoppingmall} and \textit{Restaurant} which can be found online\footnote{Website: \url{perception.i2r.a-star.edu.sg/bk_model/bk_index.html}.}.The size of each frame of \textit{Shoppingmall} is $256 \times 320$ and that of  \textit{Restaurant} is $120 \times 160$.  The total number of frames are $1000$ and $3055$ in \textit{Shoppingmall} and \textit{Restaurant}, respectively.  Each video can be represented by a matrix, where each column of the matrix is  a vectorized frame of the video.  Then, we apply each algorithm to decompose the matrix into a low rank part which represents the static background of the video and a sparse part which represents the moving objects in the video.  Since there is no ground truth for the incoherence parameter and the sparsity parameter, their values  are estimated by trial-and-error in the tests. We set $\gamma=0.7$ and $r=2$ in the decomposition of both videos, and $tol$ is set to $10^{-4}$ in the stopping criteria. All the four
algorithms can achieve desirable visual performance for the two tested videos and   we only present the decomposition results of  three selected frames for both AccAltProj with trim and without trim in Figure \ref{video_test_result}.  

Table~\ref{tab:video} contains the runtime of each algorithm. We can see that AccAltProj with and without trim are also  faster than AltProj and GD for the background subtraction experiments conducted here.   We also include the incoherence values of the output low rank matrices along the time axis.  It is worth noting that the incoherence parameter value of the  low rank output from AccAltProj with trim are smaller than that from AccAltProj without trim, which suggests the output backgrounds from AccAltProj with trim are more consistent through all the frames. Additionally, AccAltProj and AltProj have comparable output incoherence. 
\begin{figure}
\subfloat{\includegraphics[width=.19\linewidth]{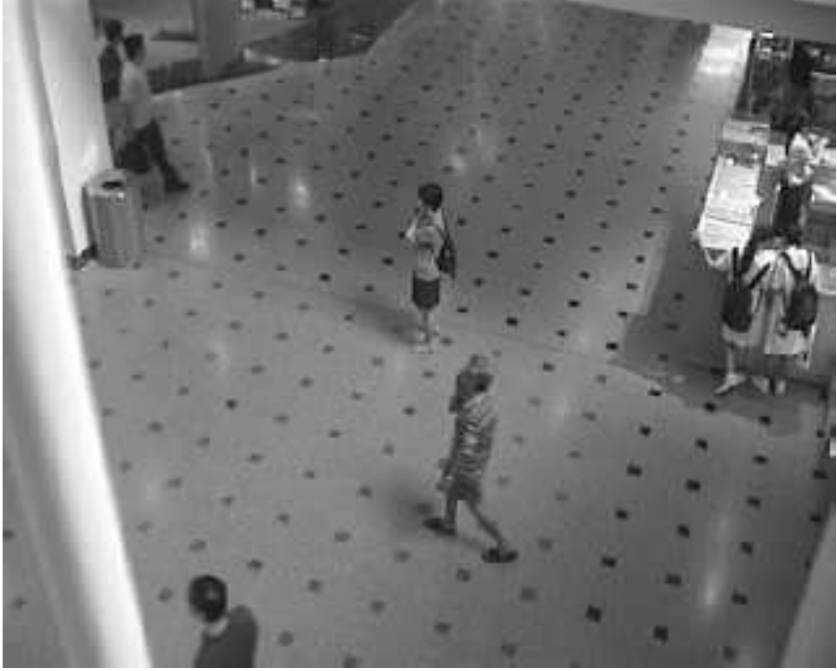}}\quad
\subfloat{\includegraphics[width=.19\linewidth]{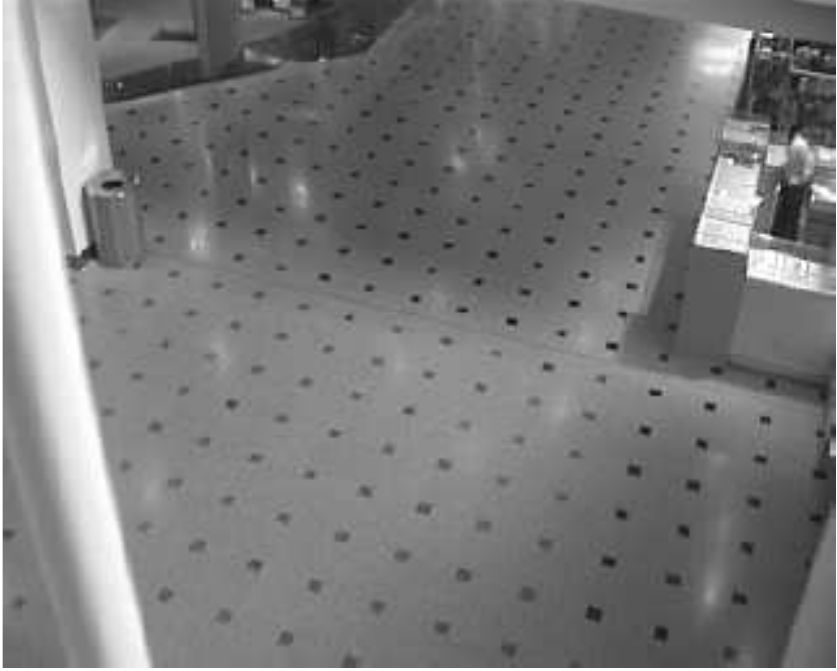}}
\subfloat{\includegraphics[width=.19\linewidth]{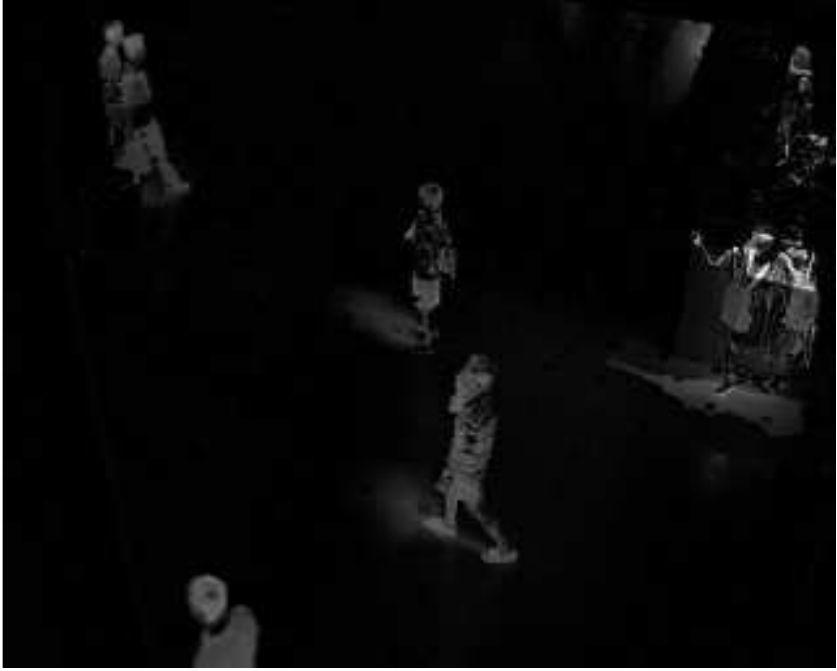}}
\subfloat{\includegraphics[width=.19\linewidth]{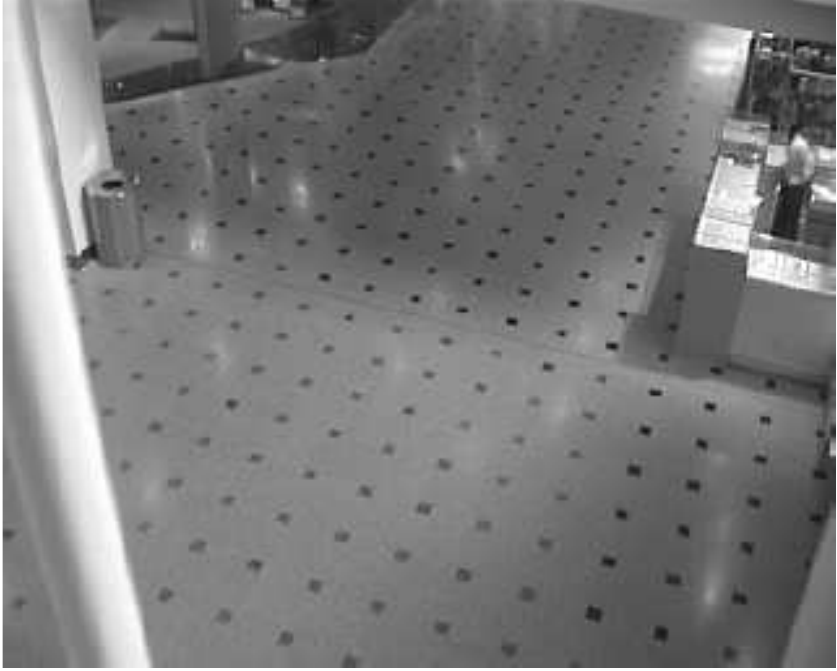}}
\subfloat{\includegraphics[width=.19\linewidth]{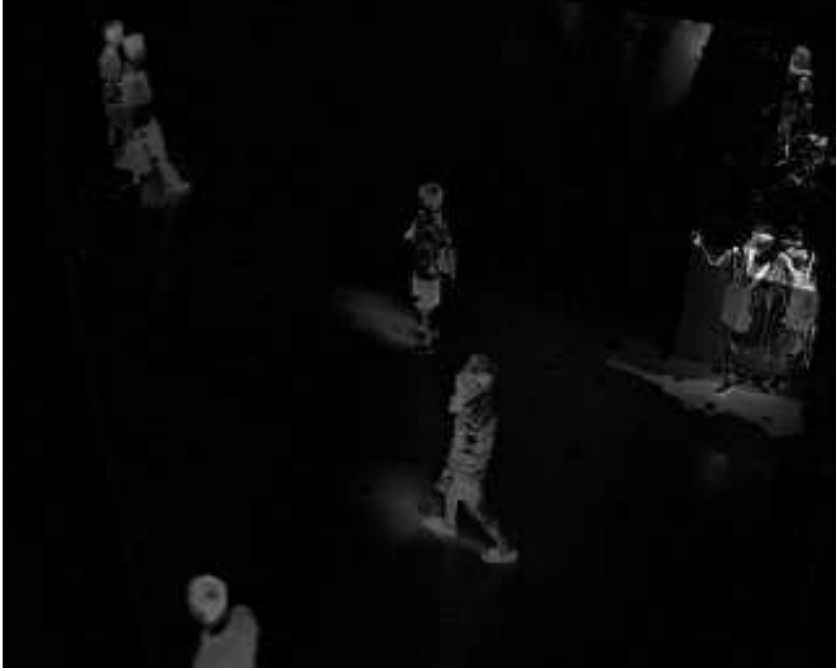}}\\
\subfloat{\includegraphics[width=.19\linewidth]{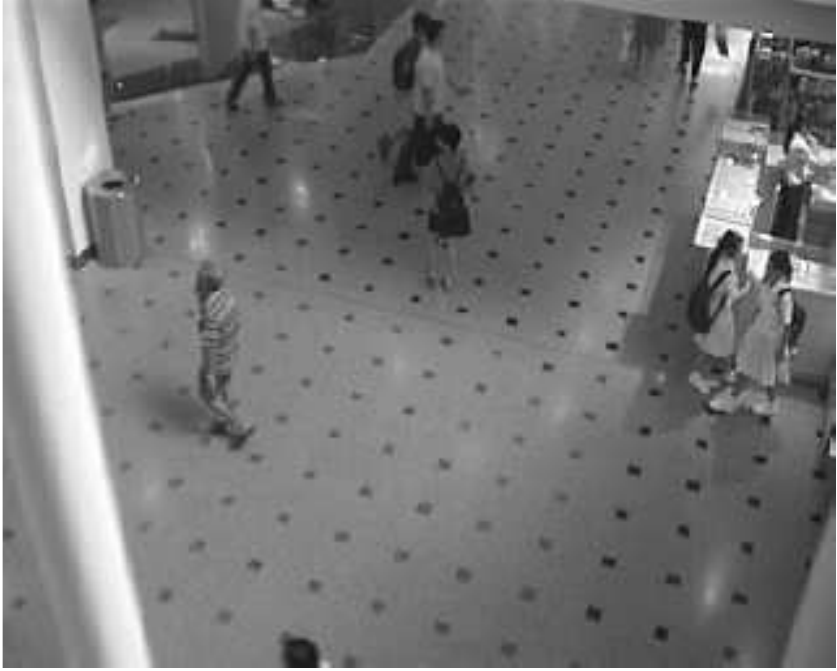}}\quad
\subfloat{\includegraphics[width=.19\linewidth]{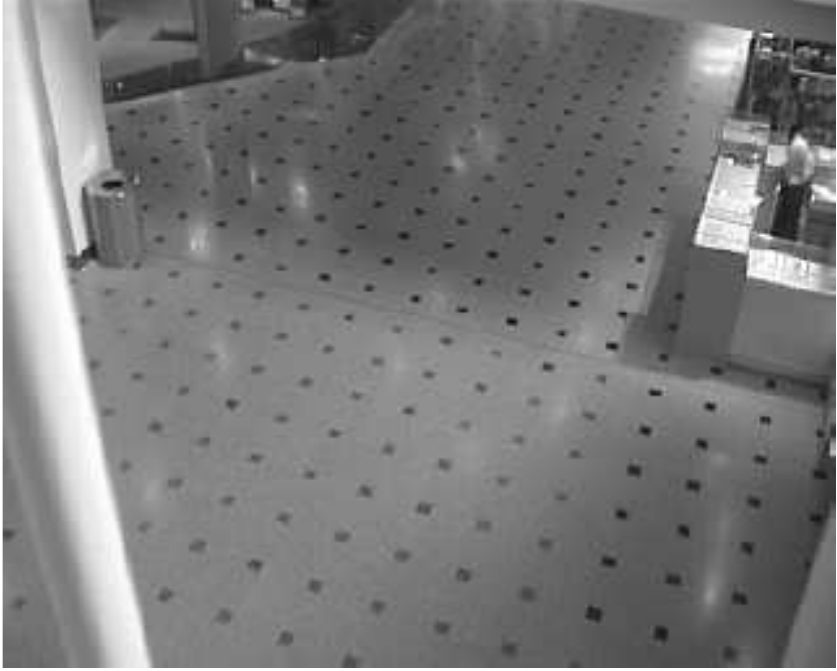}}
\subfloat{\includegraphics[width=.19\linewidth]{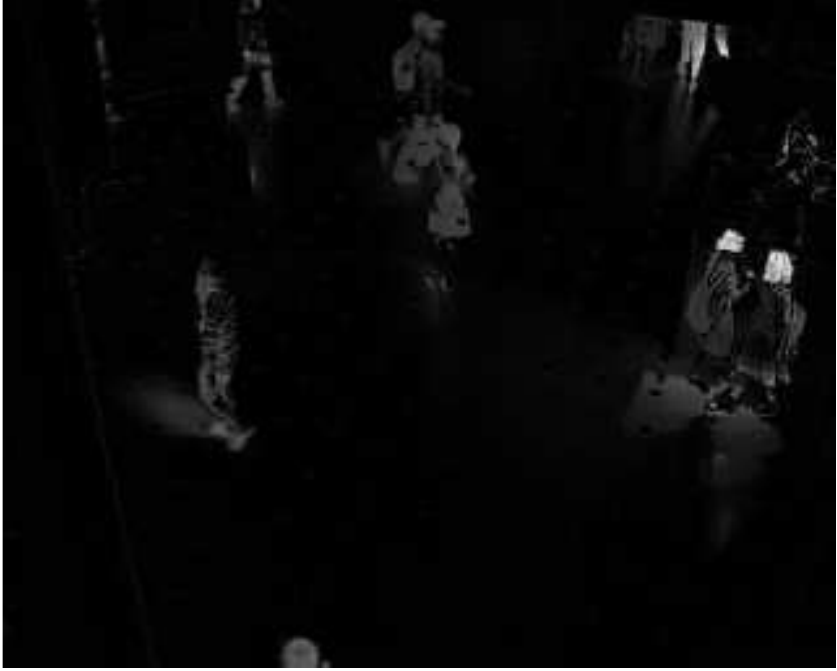}}
\subfloat{\includegraphics[width=.19\linewidth]{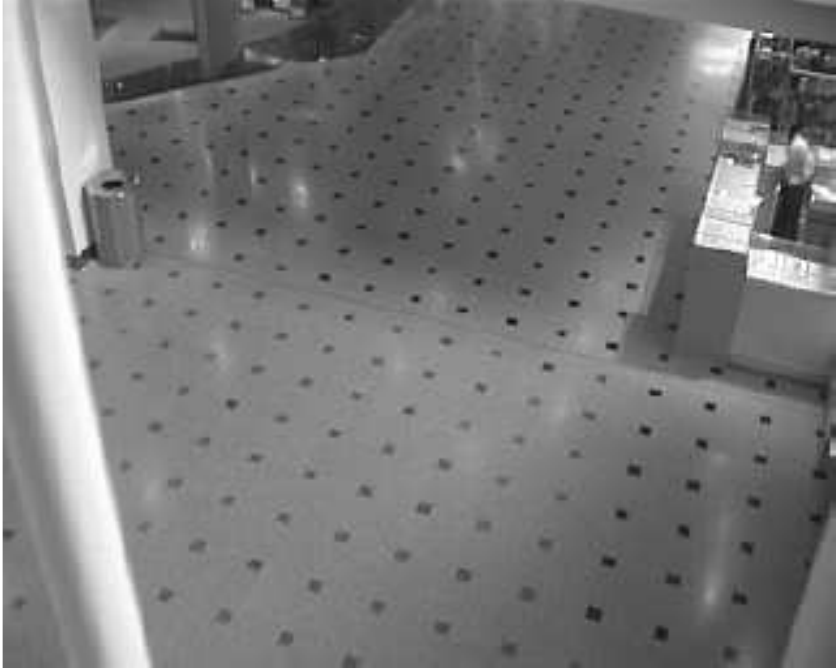}}
\subfloat{\includegraphics[width=.19\linewidth]{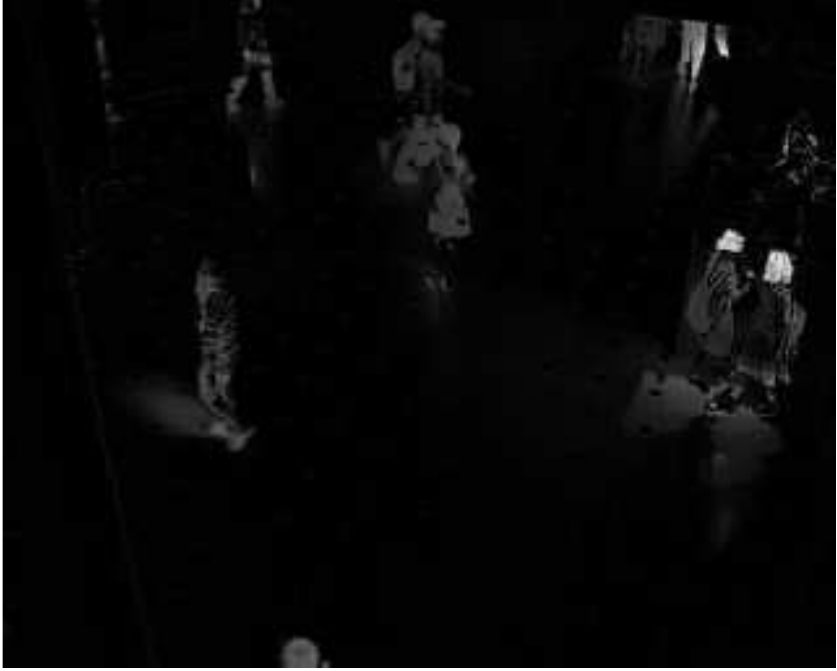}}\\
\subfloat{\includegraphics[width=.19\linewidth]{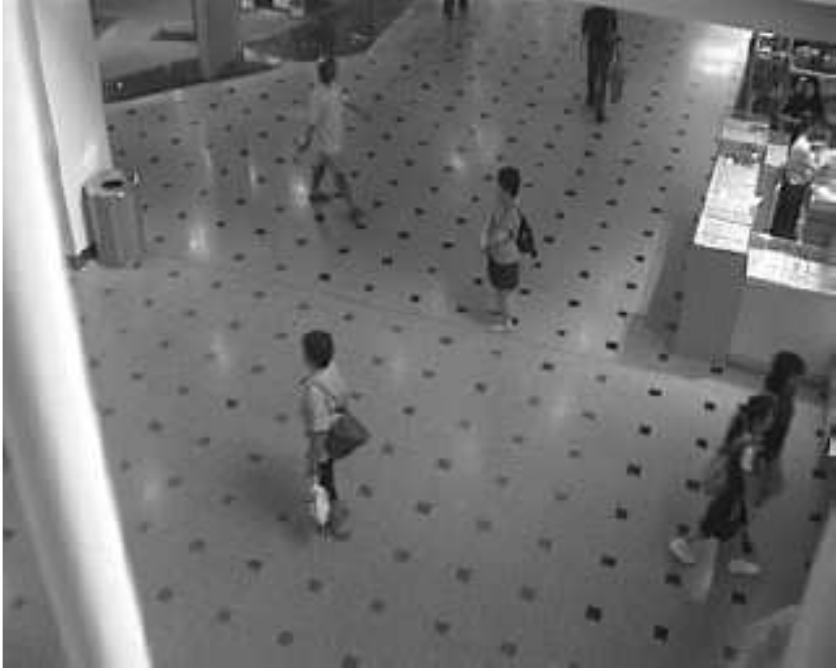}}\quad
\subfloat{\includegraphics[width=.19\linewidth]{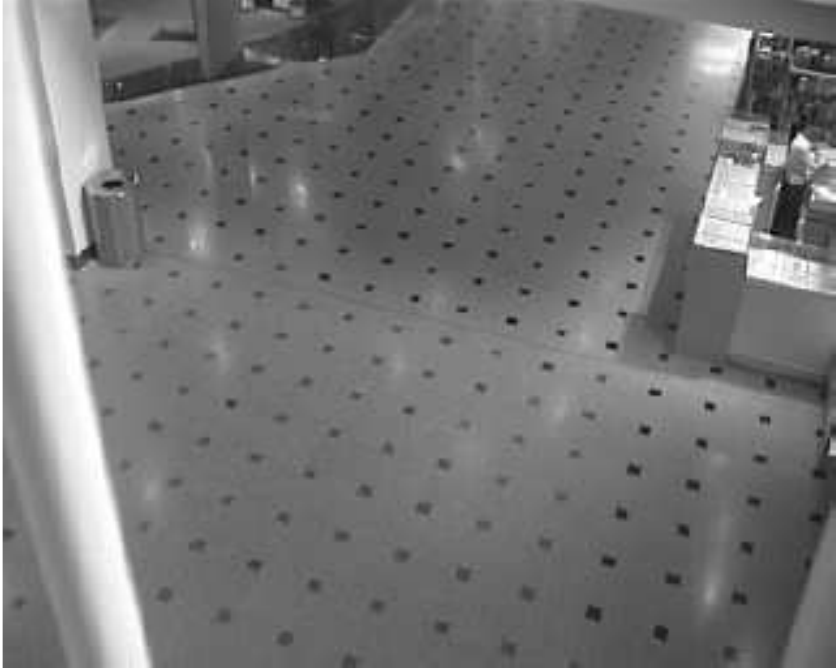}}
\subfloat{\includegraphics[width=.19\linewidth]{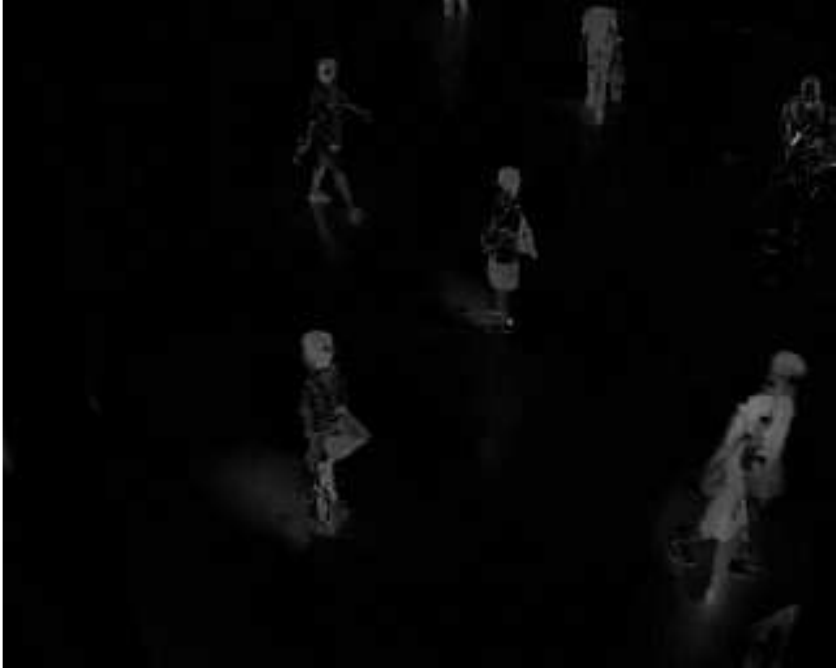}}
\subfloat{\includegraphics[width=.19\linewidth]{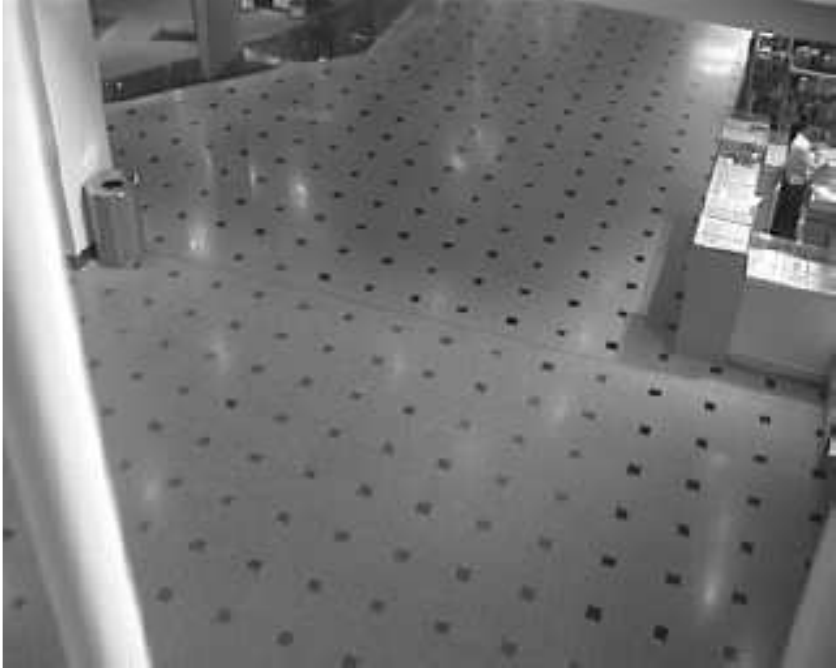}}
\subfloat{\includegraphics[width=.19\linewidth]{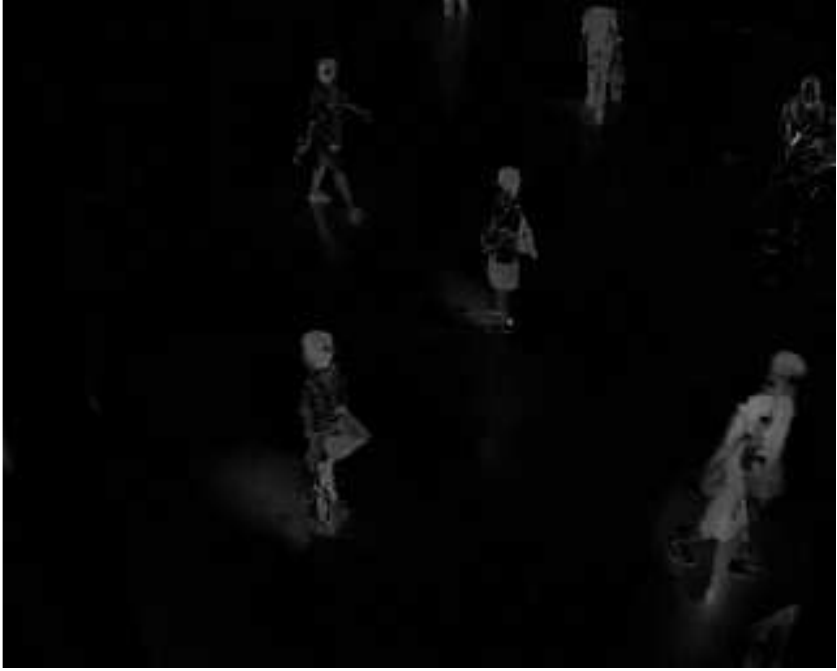}}\\
\subfloat{\includegraphics[width=.19\linewidth]{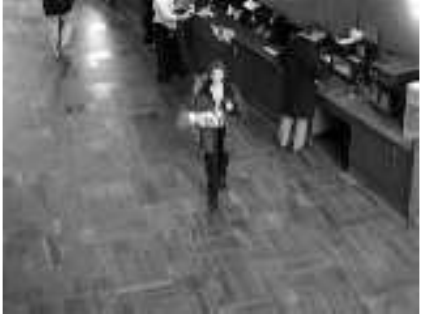}}\quad
\subfloat{\includegraphics[width=.19\linewidth]{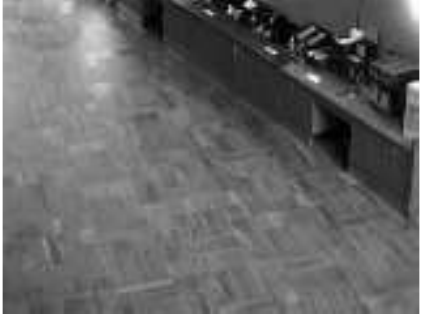}}
\subfloat{\includegraphics[width=.19\linewidth]{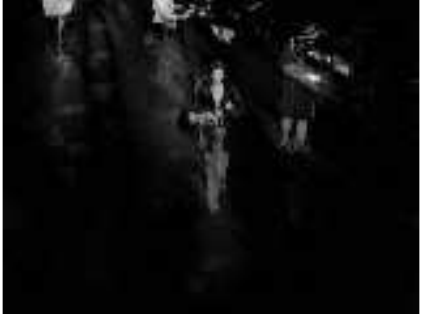}}
\subfloat{\includegraphics[width=.19\linewidth]{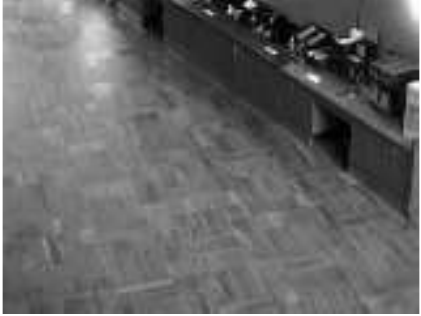}}
\subfloat{\includegraphics[width=.19\linewidth]{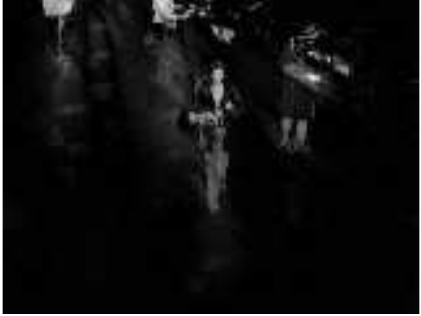}}\\
\subfloat{\includegraphics[width=.19\linewidth]{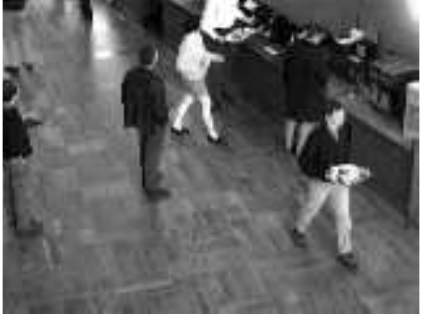}}\quad
\subfloat{\includegraphics[width=.19\linewidth]{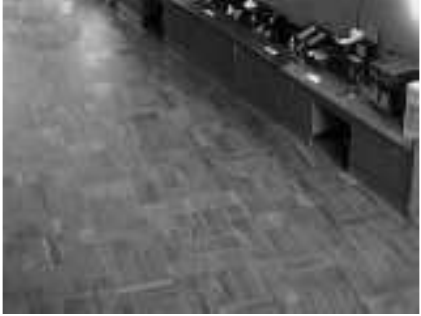}}
\subfloat{\includegraphics[width=.19\linewidth]{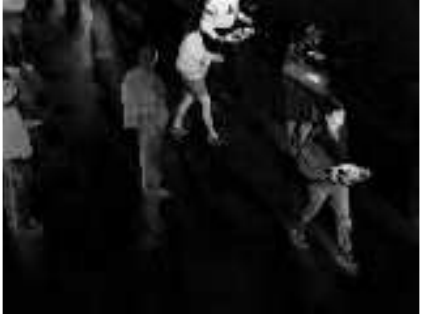}}
\subfloat{\includegraphics[width=.19\linewidth]{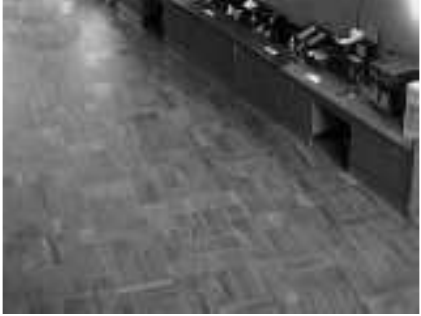}}
\subfloat{\includegraphics[width=.19\linewidth]{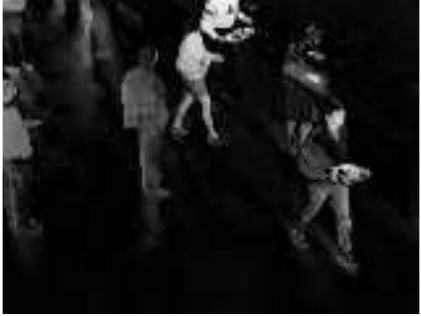}}\\
\subfloat{\includegraphics[width=.19\linewidth]{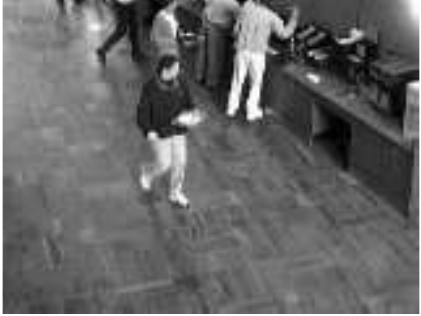}}\quad
\subfloat{\includegraphics[width=.19\linewidth]{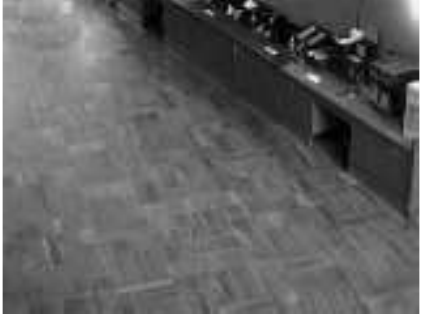}}
\subfloat{\includegraphics[width=.19\linewidth]{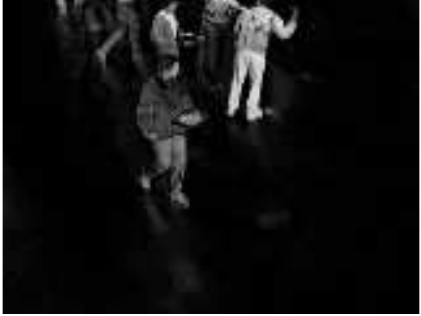}}
\subfloat{\includegraphics[width=.19\linewidth]{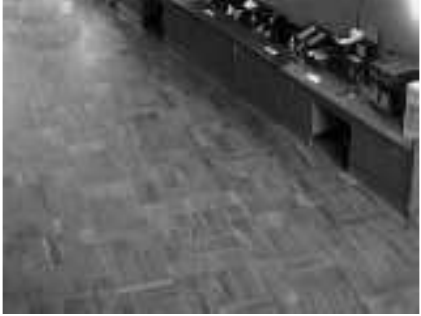}}
\subfloat{\includegraphics[width=.19\linewidth]{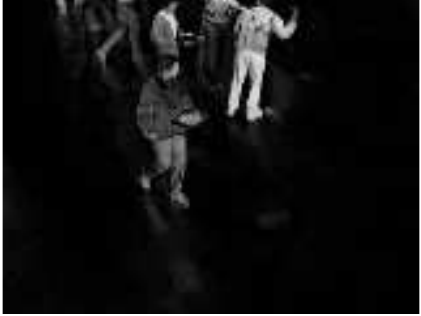}}\\
\caption{Video background subtraction: The top three rows correspond to three different frames from the video \textit{Shoppingmall}, while the bottom three rows are frames from the video \textit{Restaurant}. The first column contains the original frames,  the middle two columns are the separated background and foreground  outputs of AccAltProj with trim, and  the right two columns are the separated background and foreground  outputs of AccAltProj without trim.}\label{video_test_result}
\end{figure}

\begin{table}

\caption{Computational results for video background subtraction. Here ``S'' represents \textit{Shoppingmall}, ``R'' represents \textit{Restaurant}, and $\mu$ is the incoherence parameter of the output low rank matrices  along the time axis (i.e., among different frames).}\label{tab:video}

\centering

 \begin{tabular}{ |c||c|c|c|c|c|c|c|c| } 

\hline

~ &  \multicolumn{2}{c|}{AccAltProj w/ trim}  & \multicolumn{2}{c|}{AccAltProj w/o trim} & \multicolumn{2}{c|}{AltProj} & 
\multicolumn{2}{c|}{GD}\cr

\cline{2-9}

~                       & runtime &  $\mu$                & runtime &$\mu$  
& runtime &  $\mu$
& runtime &  $\mu$          \cr

\hhline {|=||=|=|=|=|=|=|=|=|}

S   &  $38.98s$    & $2.12$    & $38.79s$     & $2.26$ & $82.97s$ & $2.13$  &$161.1s$ & $2.85$\cr

R     &  $28.09s$    & $5.16$    & $27.94s$     & $5.25$ & $69.12s$ & $5.28$  & $107.3s$  & $6.07$\cr

\hline

\end{tabular}

\end{table}
\section{Proofs} \label{sec:proofs}
\subsection{Proof of Theorem \ref{thm:local convergence}} \label{subsec:proof of local convergence}
The proof of Theorem~\ref{thm:local convergence} follows a route established in \citep{netrapalli2014non}. Despite this, the details of the proof itself are nevertheless quite involved because there are two more operations (i.e., projection onto a tangent space and trim) in AccAltProj than in AltProj. 
Overall, the proof consists of two steps: 
\begin{itemize}
\item When $\ln\BL-\BL_k\rn_2$ and $\ln\BS-\BS_k\rn_\infty$ are sufficiently small, and $supp(\BS_k)\subset\Omega$, then $\ln\BL-\BL_{k+1}\rn_2$ deceases in some sense by a constant factor (see Lemma~\ref{lemma:Bound_of_L-L_k_2_norm}) and $\ln\BL-\BL_{k+1}\rn_\infty$ is small (see Lemma~\ref{lemma:Bound_of_L-L_k_infty_norm}).
\item When $\ln\BL-\BL_{k+1}\rn_\infty$ is sufficiently small, we can choose $\zeta_{k+1}$ such that $supp(\BS_{k+1})\subset\Omega$ and $\ln\BS-\BS_{k+1}\rn_\infty$  is small (see Lemma~\ref{lemma:Bound_of_S-S_k}).
\end{itemize}
These results will be presented in a set of lemmas. For ease of notation we define $\tau:=4\alpha\mu r\kappa$ and $\upsilon:=\tau(48\sqrt{\mu}r\kappa+\mu r)$ in the sequel. 
\begin{lemma}  \label{lemma:Weyls_inequality}
\textbf{(Weyl's inequality)} Let $\bm{A}, \bm{B}, \bm{C}\in \mathbb{R}^{n\times n}$ be the symmetric matrices such that $\bm{A} = \bm{B} + \bm{C}$. Then the inequality $$|\sigma_i^A-\sigma_i^B| \leq \|\bm{C}\|_2$$ holds for all $i$, where $\sigma_i^A$ and $\sigma_i^B$ represent the $i^{th}$ singular values of $\bm{A}$ and $\bm{B}$ respectively.
\end{lemma}
\begin{proof}
This is a well-known result and the proof can be found in many  standard textbooks, see for example \cite{bhatia2013matrix}.
\end{proof}

\begin{lemma}  \label{lemma:bound of sparse matrix}
Let $\bm{S} \in \mathbb{R}^{n\times n}$ be a symmetric sparse matrix which satisfies Assumption \nameref{assume:Sparse}. Then, the inequality
\[
\|\bm{S}\|_2 \leq \alpha n\|\bm{S}\|_\infty
\]
holds, where $\alpha$ is the sparsity level of $\bm{S}$.
\end{lemma}
\begin{proof}
The proof can be found in \citep[Lemma 4]{netrapalli2014non}.
\end{proof}

\begin{lemma}  \label{lemma:trimmed bound}
Let \textnormal{Trim} be the algorithm defined by Algorithm \ref{Algo:Trim}. If $\bm{L}_k\in \mathbb{R}^{n\times n}$ is a rank-$r$ matrix with
\begin{equation*}
\|\bm{L}_k-\bm{L}\|_2 \leq \frac{\sigma_r^L}{20\sqrt{r}} ,
\end{equation*}
then the trim output  with the level $\sqrt{\frac{\mu r}{n}}$ 
satisfies
\begin{equation}  \label{eq:trim bound F-norm}
\|\widetilde{\bm{L}}_{k}-\bm{L}\|_F\leq8\kappa \|\bm{L}_{k}-\bm{L}\|_F,
\end{equation}
\begin{equation}  \label{eq:trimmed incoh }
\max_i \|\bm{e}_i^T \widetilde{\bm{U}}_k\|_2\leq \frac{10}{9} \sqrt{\frac{\mu r}{n}}, \quad\textnormal{and}\quad \max_j \|\bm{e}_j^T \widetilde{\bm{V}}_k\|_2\leq \frac{10}{9}\sqrt{\frac{\mu r}{n}},
\end{equation}
where  $\widetilde{\bm{L}}_k=\widetilde{\bm{U}}_k\widetilde{\bm{\Sigma}}_k\widetilde{\bm{V}}_k^T$ is the SVD of $\widetilde{\bm{L}}_k$. Furthermore, it follows that
\begin{equation}  \label{eq:trim bound 2-norm}
\|\widetilde{\bm{L}}_{k}-\bm{L}\|_2\leq8\sqrt{2r}\kappa \|\bm{L}_{k}-\bm{L}\|_2.
\end{equation}
\end{lemma}

\begin{proof}
Since both $\bm{L}$ and $\bm{L}_k$ are rank-$r$ matrices, $\bm{L}_{k}-\bm{L}$ is  rank at most $2r$. So
\[
\|\bm{L}_k-\bm{L}\|_F\leq\sqrt{2r}\|\bm{L}_k-\bm{L}\|_2 \leq \sqrt{2r}\frac{\sigma_r^L}{20\sqrt{r}}=\frac{\sigma_r^L}{10\sqrt{2}} .
\]
Then, the first two parts  of the lemma, i.e., (\ref{eq:trim bound F-norm}) and (\ref{eq:trimmed incoh }), follow from  \citep[Lemma 4.10]{wei2016guarantees_completion}. Noting that $\|\widetilde{\bm{L}}_{k}-\bm{L}\|_2\leq \|\widetilde{\bm{L}}_{k}-\bm{L}\|_F$,
\eqref{eq:trim bound 2-norm} follows immediately. 
\end{proof}

\begin{lemma}  \label{lemma: secound_order_bound_I-P}
Let $\bm{L}=\bm{U}\bm{\Sigma}\bm{V}^T$ and $\widetilde{\bm{L}}_k=\widetilde{\bm{U}}_k\widetilde{\bm{\Sigma}}_k\widetilde{\bm{V}}_k^T$ be the SVD of two rank-$r$ matrices,
then
\begin{equation}  \label{eq:secound_order_bound_I-P 1}
\|\bm{U}\bm{U}^T-\widetilde{\bm{U}}\widetilde{\bm{U}}^T\|_2\leq \frac{\|\widetilde{\bm{L}}_{k}-\bm{L}\|_2}{\sigma_r^L}, \qquad \|\bm{V}\bm{V}^T-\widetilde{\bm{V}}\widetilde{\bm{V}}^T\|_2\leq \frac{\|\widetilde{\bm{L}}_{k}-\bm{L}\|_2}{\sigma_r^L},
\end{equation}
and
\begin{equation}   \label{eq:secound_order_bound_I-P 2}
\|(\mathcal{I}-\mathcal{P}_{\widetilde{T}_k}) \bm{L}\|_2\leq\frac{\|\widetilde{\bm{L}}_{k}-\bm{L}\|_2^2}{\sigma_r^L}.
\end{equation}
\end{lemma}

\begin{proof}
The proof of (\ref{eq:secound_order_bound_I-P 1}) can be found in \citep[Lemma 4.2]{wei2016guarantees_recovery}. The Frobenius norm version of (\ref{eq:secound_order_bound_I-P 2}) can also be found in  \citep{wei2016guarantees_completion, wei2016guarantees_recovery}. Here we only need to prove the spectral norm version, i.e., (\ref{eq:secound_order_bound_I-P 2}). 
Since $\bm{L}=\bm{U}\bm{U}^T\bm{L}$ and $\widetilde{\bm{L}}_k(\bm{I}-\widetilde{\bm{V}}_k\widetilde{\bm{V}}_k^T)=\bm{0}$,  we have
\begin{align*}
\|(\mathcal{I}-\mathcal{P}_{\widetilde{T}_k})\bm{L}\|_2&=\|(\bm{I}-\widetilde{\bm{U}}_k\widetilde{\bm{U}}_k^T) \bm{L}(\bm{I}-\widetilde{\bm{V}}_k\widetilde{\bm{V}}_k^T)\|_2 \cr
												      &=\|(\bm{I}-\widetilde{\bm{U}}_k\widetilde{\bm{U}}_k^T) \bm{U}\bm{U}^T\bm{L}(\bm{I}-\widetilde{\bm{V}}_k\widetilde{\bm{V}}_k^T)\|_2 \cr
												      &=\|(\bm{U}\bm{U}^T-\widetilde{\bm{U}}_k\widetilde{\bm{U}}_k^T) \bm{U}\bm{U}^T\bm{L}(\bm{I}-\widetilde{\bm{V}}_k\widetilde{\bm{V}}_k^T)\|_2 \cr
												      &=\|(\bm{U}\bm{U}^T-\widetilde{\bm{U}}_k\widetilde{\bm{U}}_k^T) \bm{L}(\bm{I}-\widetilde{\bm{V}}_k\widetilde{\bm{V}}_k^T)\|_2 \cr
												      &=\|(\bm{U}\bm{U}^T-\widetilde{\bm{U}}_k\widetilde{\bm{U}}_k^T) (\bm{L}-\widetilde{\bm{L}}_k)(\bm{I}-\widetilde{\bm{V}}_k\widetilde{\bm{V}}_k^T)\|_2 \cr
												      &\leq\|(\bm{U}\bm{U}^T-\widetilde{\bm{U}}_k\widetilde{\bm{U}}_k^T)\|_2 \|(\bm{L}-\widetilde{\bm{L}}_k)\|_2 \|(\bm{I}-\widetilde{\bm{V}}_k\widetilde{\bm{V}}_k^T)\|_2 \cr
												      &\leq\frac{\|\widetilde{\bm{L}}_{k}-\bm{L}\|_2^2}{\sigma_r^L},
\end{align*}
where the last inequality follows from (\ref{eq:secound_order_bound_I-P 1}).
\end{proof}

\begin{lemma}  \label{lemma:error_of_projected_max_norm}
Let  $\bm{S}\in\mathbb{R}^{n\times n}$ be a symmetric matrix satisfying Assumption~\nameref{assume:Sparse}. Let $\widetilde{\bm{L}}_k\in\mathbb{R}^{n\times n}$ be a  rank-$r$ matrix with $\frac{100}{81}\mu$-incoherence. That is,
\begin{equation*}
\max_i \|\bm{e}_i^T \widetilde{\bm{U}}_k\|_2\leq \frac{10}{9} \sqrt{\frac{\mu r}{n}} \quad\textnormal{and}\quad \max_j \|\bm{e}_j^T \widetilde{\bm{V}}_k\|_2\leq \frac{10}{9}\sqrt{\frac{\mu r}{n}},
\end{equation*}
where $\widetilde{\bm{L}}=\widetilde{\bm{U}}_k\widetilde{\bm{\Sigma}}_k\widetilde{\bm{V}}_k^T$ is the SVD of $\widetilde{\bm{L}}_k$. 
 If $supp(\bm{S}_k)\subset\Omega$, then
\begin{equation}
\|\mathcal{P}_{\widetilde{T}_k} (\bm{S}-\bm{S}_k)\|_\infty \leq 4\alpha \mu r \|\bm{S}-\bm{S}_k\|_\infty .
\end{equation}
\end{lemma}
\begin{proof}
By the incoherence assumption of $\widetilde{\bm{L}}_{k}$ and the sparsity assumption of $\bm{S}-\bm{S}_k$, we have
\begin{align*}
[\mathcal{P}_{\widetilde{T}_k}(\bm{S}-\bm{S}_k)]_{ab} &=\langle \mathcal{P}_{\widetilde{T}_k}(\bm{S}-\bm{S}_k),\bm{e}_a\bm{e}_b^T \rangle \cr
                           &=\langle \bm{S}-\bm{S}_k,\mathcal{P}_{\widetilde{T}_k}(\bm{e}_a\bm{e}_b^T) \rangle  \cr
                           &=\langle \bm{S}-\bm{S}_k,\widetilde{\bm{U}}_k\widetilde{\bm{U}}_k^T\bm{e}_a\bm{e}_b^T   + \bm{e}_a\bm{e}_b^T\widetilde{\bm{V}}_k\widetilde{\bm{V}}_k^T   -  \widetilde{\bm{U}}_k\widetilde{\bm{U}}_k^T\bm{e}_a\bm{e}_b^T \widetilde{\bm{V}}_k\widetilde{\bm{V}}_k^T \rangle  \cr
                           &=  \langle (\bm{S}-\bm{S}_k)\bm{e}_b, \widetilde{\bm{U}}_k\widetilde{\bm{U}}_k^T \bm{e}_a\rangle + \langle \bm{e}_a^T(\bm{S}-\bm{S}_k),\bm{e}_b^T \widetilde{\bm{V}}_k\widetilde{\bm{V}}_k^T\rangle  - \langle \bm{S}-\bm{S}_k,\widetilde{\bm{U}}_k\widetilde{\bm{U}}_k^T\bm{e}_a\bm{e}_b^T \widetilde{\bm{V}}_k\widetilde{\bm{V}}_k^T \rangle  \cr
                           &\leq \|\bm{S}-\bm{S}_k\|_\infty\left(\sum_{i|(i,b)\in\Omega} |\bm{e}_i^T \widetilde{\bm{U}}_k\widetilde{\bm{U}}_k^T\bm{e}_a| + \sum_{j|(a,j)\in\Omega} |\bm{e}_b^T \widetilde{\bm{V}}_k\widetilde{\bm{V}}_k^T\bm{e}_j|\right)\\
                           &\quad+  \|\bm{S}-\bm{S}_k\|_2~\|\widetilde{\bm{U}}_k\widetilde{\bm{U}}_k^T\bm{e}_a\bm{e}_b^T \widetilde{\bm{V}}_k\widetilde{\bm{V}}_k^T\|_\ast       \cr
                           &\leq 2\alpha n\frac{100\mu r}{81n}\|\bm{S}-\bm{S}_k\|_\infty  + \alpha n \|\bm{S}-\bm{S}_k\|_\infty\|\widetilde{\bm{U}}_k\widetilde{\bm{U}}_k^T\bm{e}_a\bm{e}_b^T \widetilde{\bm{V}}_k\widetilde{\bm{V}}_k^T\|_F \cr
                           &\leq \frac{200}{81}\alpha \mu r\|\bm{S}-\bm{S}_k\|_\infty  + \alpha n \frac{\mu r}{n} \|\bm{S}-\bm{S}_k\|_\infty  \cr
                           &= 4\alpha \mu r \|\bm{S}-\bm{S}_k\|_\infty,
\end{align*}
where the first inequality uses H\"older's inequality and the second inequality uses Lemma \ref{lemma:bound of sparse matrix}. We also use the fact $\widetilde{\bm{U}}_k\widetilde{\bm{U}}_k^T\bm{e}_a\bm{e}_b^T\widetilde{\bm{V}}_k\widetilde{\bm{V}}_k^T$ is a rank-$1$  matrix to bound its nuclear norm. 
\end{proof}


\begin{lemma}   \label{lemma:P_T X 2-norm ineq}
Under the symmetric setting, i.e., $\bm{U}\bm{U}^T=\bm{V}\bm{V}^T$ where $\bm{U}\in\R^{n\times r}$ and $\bm{V}\in\R^{n\times r}$ are two orthogonal matrices,  we have
\[
\|\mathcal{P}_T \bm{Z}\|_2 \leq \sqrt{\frac{4}{3}}\|\bm{Z}\|_2
\]
for any symmetric matrix $\bm{Z}\in\mathbb{R}^{n\times n}$. Moreover, the upper bound  is tight.
\end{lemma}

\begin{proof}
First notice that $$\mathcal{P}_T \bm{Z} = \bm{U}\bm{U}^T\bm{Z}+ \bm{Z}\bm{U}\bm{U}^T-\bm{U}\bm{U}^T\bm{Z}\bm{U}\bm{U}^T$$ is symmetric. 
Let $\bm{y}\in\mathbb{R}^n$ be 
a unit vector such that $\|\mathcal{P}_{T}\bm{Z}\|_2=|\bm{y}^T(\mathcal{P}_T\bm{Z})\bm{y}|$.
Denote $\bm{y}_1=\bm{U}\bm{U}^T\bm{y}$ and $\bm{y}_2=(\bm{I}-\bm{U}\bm{U}^T)\bm{y}$. Then,
\begin{align*}
\|\mathcal{P}_T \bm{Z}\|_2&=|\bm{y}^T(\mathcal{P}_T \bm{Z}) \bm{y}| \cr
           &=|\bm{y}_1^T\bm{Z}\bm{y}+\bm{y}^T\bm{Z}\bm{y}_1-\bm{y}_1^T\bm{Z}\bm{y}_1| \cr
           &=|\bm{y}_1^T\bm{Z}\bm{y}_1+\bm{y}_1^T\bm{Z}\bm{y}_2+\bm{y}_2^T\bm{Z}\bm{y}_1| \cr
           &=|\langle \bm{y}_1\bm{y}_1^T+\bm{y}_1\bm{y}_2^T+\bm{y}_2\bm{y}_1^T , \bm{Z} \rangle| \cr
           &\leq  \| \bm{y}_1\bm{y}_1^T+\bm{y}_1\bm{y}_2^T+\bm{y}_2\bm{y}_1^T\|_\ast \|\bm{Z}\|_2.
\end{align*}
Let $a=\|\bm{y}_1\|_2^2$. Since $\bm{y}_1 \perp \bm{y}_2$, we have $\|\bm{y}_1\|_2^2+\|\bm{y}_2\|_2^2=1$, which  implies $\|\bm{y}_2\|_2^2=1-a$ and
\begin{align*}
\bm{y}_1\bm{y}_1^T+\bm{y}_1\bm{y}_2^T+\bm{y}_2\bm{y}_1^T &= \begin{bmatrix} \bm{y}_1 & \bm{y}_2\end{bmatrix} \begin{bmatrix} 1  & 1 \\ 1 & 0\end{bmatrix} \begin{bmatrix} \bm{y}_1 & \bm{y}_2\end{bmatrix}^T \cr
                           &= \begin{bmatrix} \frac{\bm{y}_1}{\sqrt{a}} & \frac{\bm{y}_2}{\sqrt{1-a}}\end{bmatrix} \begin{bmatrix} a  & \sqrt{a(1-a)} \\ \sqrt{a(1-a)} & 0\end{bmatrix} \begin{bmatrix} \frac{\bm{y}_1}{\sqrt{a}} & \frac{\bm{y}_2}{\sqrt{1-a}}\end{bmatrix}^T.
\end{align*}
Since $\begin{bmatrix} \frac{\bm{y}_1}{\sqrt{a}} & \frac{\bm{y}_2}{\sqrt{1-a}}\end{bmatrix}$ is an orthogonal matrix,  one has
\begin{align*}
\|\bm{y}_1\bm{y}_1^T+\bm{y}_1\bm{y}_2^T+\bm{y}_2\bm{y}_1^T\|_\ast &= \norm { \begin{bmatrix} a  & \sqrt{a(1-a)} \\ \sqrt{a(1-a)} & 0\end{bmatrix} }_\ast  \cr
                                    &=\sqrt{a^2+4a(1-a)}  \cr
                                    &=\sqrt{\frac{4}{3}-3\left(a-\frac{2}{3}\right)^2}   \cr
                                    &\leq  \sqrt{\frac{4}{3}},
\end{align*}
which complete the proof for the upper bound.

To show the tightness of the bound, let $\bm{U}=\bm{V}= \begin{bmatrix} 1 \\ 0\end{bmatrix}$ and $\bm{Z}=\begin{bmatrix} 1  & \sqrt{2} \\ \sqrt{2} & -1\end{bmatrix}$. It can be easily verified  that  $\|\mathcal{P}_T \bm{\bm{Z}}\|_2 = \sqrt{\frac{4}{3}}\|\bm{\bm{Z}}\|_2$.
\end{proof}

\begin{lemma} \label{lemma:bound_power_vector_norm_with_incoherence}
Let $\bm{U}\in\mathbb{R}^{n\times r}$ be an orthogonal matrix with $\mu$-incoherence, i.e., $\|\bm{e}_i^T\bm{U}\|_2\leq\sqrt{\frac{\mu r}{n}}$ for all $i$. Then, for any $\bm{Z}\in\mathbb{R}^{n\times n}$, the inequality 
\[
\|\bm{e}_i^T\bm{Z}^a\bm{U}\|_2\leq \max_l\sqrt{\frac{\mu r}{n}}(\sqrt{n}\|\bm{e}_l^T\bm{Z}\|_2)^a
\]
holds for all $i$ and $a\geq 0$.
\end{lemma}
\begin{proof}
This proof is done by mathematical induction.\\
\textbf{Base case:} When $a=0$, $\|\bm{e}_i^T\bm{U}\|\leq\sqrt{\frac{\mu r}{n}}$ is satisfied following from the assumption. \\
\textbf{Induction Hypothesis:} $\|\bm{e}_i^T(\bm{Z})^a\bm{U}\|_2\leq \max_l\sqrt{\frac{\mu r}{n}}(\sqrt{n}\|\bm{e}_l^T\bm{Z}\|_2)^a $ for all $i$ at the $a^{th}$ power. \\
\textbf{Induction Step:} We have
\begin{align*}
\|\bm{e}_i^T\bm{Z}^{a+1}\bm{U}\|_2^2 &= \|\bm{e}_i^T\bm{Z}\bm{Z}^a\bm{U}\|_2^2 \cr
                        &=\sum_j\left(\sum_k [\bm{Z}]_{ik}[\bm{Z}^a\bm{U}]_{kj}\right)^2 \cr
                        &=\sum_{k_1k_2} [\bm{Z}]_{ik_1}[\bm{Z}]_{ik_2}\sum_j [\bm{Z}^a\bm{U}]_{k_1j} [\bm{Z}^a\bm{U}]_{k_2j}  \cr
                        &=\sum_{k_1k_2} [\bm{Z}]_{ik_1}[\bm{Z}]_{ik_2} \langle \bm{e}_{k_1}^T\bm{Z}^a\bm{U},\bm{e}_{k_2}^T\bm{Z}^a\bm{U} \rangle \cr
                        &\leq \sum_{k_1k_2} |[\bm{Z}]_{ik_1}[\bm{Z}]_{ik_2}|\|\bm{e}_{k_1}^T\bm{Z}^a\bm{U}\|_2~\|\bm{e}_{k_2}^T\bm{Z}^a\bm{U}\|_2  \cr
                        &\leq \max_l\frac{\mu r}{n} (\sqrt{n}\|\bm{e}_l^T\bm{Z}\|_2)^{2a}\sum_{k_1k_2} |[\bm{Z}]_{ik_1}[\bm{Z}]_{ik_2}|  \cr
                        &\leq \max_l\frac{\mu r}{n} (\sqrt{n}\|\bm{e}_l^T\bm{Z}\|_2)^{2a}(\sqrt{n}\|\bm{e}_i^T\bm{Z}\|_2)^2  \cr
                        &\leq \max_l\frac{\mu r}{n} (\sqrt{n}\|\bm{e}_l^T\bm{Z}\|_2)^{2a+2}.
\end{align*}
The proof is complete by taking a square root from both sides. 
\end{proof}

\begin{lemma} \label{lemma:norm_of_Z}
Let $\bm{L}\in\mathbb{R}^{n\times n}$ and $\bm{S}\in\mathbb{R}^{n\times n}$ be two symmetric matrices satisfying Assumptions \nameref{assume:Inco} and \nameref{assume:Sparse}, respectively. 
Let $\widetilde{\bm{L}}_k\in\mathbb{R}^{n\times n}$  be the trim output of $\bm{L}_k$. If
\[
\|\bm{L}-\bm{L}_k\|_2 \leq 8\alpha \mu r \gamma^k\sigma_1^L,\quad
\|\bm{S}-\bm{S}_k\|_\infty \leq \frac{\mu r}{n} \gamma^k\sigma_1^L,\textnormal{\ and }
supp(\bm{S}_k)\subset \Omega,
\]
then  
\begin{equation}  \label{eq:norm_of_Z 1}
\|(\mathcal{P}_{\widetilde{T}_k}-\mathcal{I})\bm{L}+\mathcal{P}_{\widetilde{T}_k}(\bm{S}-\bm{S}_k)\|_2 \leq \tau\gamma^{k+1}\sigma_r^L
\end{equation}
and
\begin{equation} \label{eq:norm_of_Z 2}
\max_l\sqrt{n}\|\bm{e}_l^T[(\mathcal{P}_{\widetilde{T}_k}-\mathcal{I})\bm{L}+\mathcal{P}_{\widetilde{T}_k}(\bm{S}-\bm{S}_k)]\|_2 \leq \upsilon\gamma^{k}\sigma_r^L
\end{equation}
hold for all $k\geq 0$, provided $1>\gamma\geq512\tau r \kappa^2+\frac{1}{\sqrt{12}}$. Here recall that $\tau=4\alpha \mu r\kappa$ and $\upsilon=\tau(48\sqrt{\mu}r\kappa+\mu r)$.
\end{lemma}
\begin{proof}
For all $k\geq 0$, we get
\begin{align*}
\|(\mathcal{P}_{\widetilde{T}_k}-\mathcal{I})\bm{L}+\mathcal{P}_{\widetilde{T}_k}(\bm{S}-\bm{S}_k)\|_2 &\leq \|(\mathcal{P}_{\widetilde{T}_k}-\mathcal{I})\bm{L}\|_2+   \| \mathcal{P}_{\widetilde{T}_k} (\bm{S}-\bm{S}_k)\|_2\cr
                                                          &\leq \frac{\|\bm{L}-\widetilde{\bm{L}}_k\|_2^2}{\sigma_r^L} + \sqrt{\frac{4}{3}}\|\bm{S}-\bm{S}_k\|_2  \cr
                                                          &\leq \frac{(8\sqrt{2r}\kappa)^2\|\bm{L}-\bm{L}_k\|_2^2}{\sigma_r^L} + \sqrt{\frac{4}{3}}\alpha n\|\bm{S}-\bm{S}_k\|_\infty  \KW{\alpha \lesssim\frac{1}{\mu r^{3/2}\kappa}}\cr
                                                          &\leq 128\cdot 8\alpha \mu r^2 \kappa^3\|\bm{L}-\bm{L}_k\|_2+ \sqrt{\frac{4}{3}}\alpha n\|\bm{S}-\bm{S}_k\|_\infty  \cr
                                                          &\leq \left( 512 \tau r \kappa^2 + \frac{1}{4}\sqrt{\frac{4}{3}}\right) 4\alpha \mu r\gamma^k\sigma_1^L \cr
                                                          &\leq 4\alpha \mu r\gamma^{k+1}\sigma_1^L, \KW{\alpha \lesssim\frac{1}{\mu r^2\kappa^3}} \cr
                                                          &= \tau\gamma^{k+1}\sigma_r^L
\end{align*}
where the second inequality uses Lemma \ref{lemma: secound_order_bound_I-P} and \ref{lemma:P_T X 2-norm ineq}, the third inequality uses Lemma \ref{lemma:bound of sparse matrix} and \ref{lemma:trimmed bound}, the fourth inequality follows from $\frac{\|\bm{L}-\bm{L}_k\|_2}{\sigma_r^L}\leq 8\alpha \mu r \kappa$, and the last inequality uses the bound of $\gamma$. 

To compute the bound of $\max_l \sqrt{n}\|\bm{e}_l^T[(\mathcal{P}_{\widetilde{T}_k}-\mathcal{I})\bm{L}+\mathcal{P}_{\widetilde{T}_k}(\bm{S}-\bm{S}_k)]\|_2$, first note that
\begin{align*}
\max_l \|\bm{e}_l^T(\mathcal{I}-\mathcal{P}_{\widetilde{T}_k})\bm{L}\|_2 &= \max_l \|\bm{e}_l^T(\bm{U}\bm{U}^T-\widetilde{\bm{U}}_k\widetilde{\bm{U}}^T_k)(\bm{L}-\widetilde{\bm{L}}_k)(\bm{I}-\widetilde{\bm{U}}_k\widetilde{\bm{U}}^T_k)\|_2  \cr
								 &\leq \max_l \|\bm{e}_l^T(\bm{U}\bm{U}^T-\widetilde{\bm{U}}_k\widetilde{\bm{U}}^T_k)\|_2\|\bm{L}-\widetilde{\bm{L}}_k\|_2\|\bm{I}-\widetilde{\bm{U}}_k\widetilde{\bm{U}}^T_k\|_2 \cr
								 &\leq \left(\frac{19}{9}\sqrt{\frac{\mu r}{n}}\right)\|\bm{L}-\widetilde{\bm{L}}_k\|_2,
\end{align*}
where the last inequality follows from the fact $\bm{L}$ is $\mu$-incoherent and $\widetilde{\bm{L}}_k$ is $\frac{100}{81}\mu$-incoherent.
Hence, for all $k\geq 0$, we have
\begin{align*}
\max_l\sqrt{n}\|\bm{e}_l^T((\mathcal{P}_{\widetilde{T}_k}-\mathcal{I})\bm{L}+\mathcal{P}_{\widetilde{T}_k}(\bm{S}-\bm{S}_k))\|_2 &\leq \max_l\sqrt{n}\|\bm{e}_l^T(\mathcal{I}-\mathcal{P}_{\widetilde{T}_k})\bm{L}\|_2+\sqrt{n}\|\bm{e}_l^T\mathcal{P}_{\widetilde{T}_k}(\bm{S}-\bm{S}_k)\|_2 \cr
							&\leq \frac{19\sqrt{n}}{9}\sqrt{\frac{\mu r}{n}}\|\bm{L}-\widetilde{\bm{L}}_k\|_2 + n\|\mathcal{P}_{\widetilde{T}_k}(\bm{S}-\bm{S}_k)\|_\infty\cr
							&\leq \frac{19}{9}8\sqrt{2\mu}r\kappa\|\bm{L}-\bm{L}_k\|_2 + 4n\alpha \mu r\|\bm{S}-\bm{S}_k\|_\infty\cr
							&\leq 24\sqrt{\mu}r\kappa\cdot 8\alpha \mu r \gamma^k\sigma_1^L + 4n\alpha \mu r\cdot\frac{\mu r}{n} \gamma^k\sigma_1^L \cr
							&= \upsilon\gamma^{k}\sigma_r^L,  \KW{\alpha \lesssim\frac{1}{u^{3/2}r^2\kappa}} \KW{, \alpha \lesssim\frac{1}{\mu^2r^2}}
\end{align*}
where the third inequality uses Lemma \ref{lemma:trimmed bound} and \ref{lemma:error_of_projected_max_norm}. 
\end{proof}

\begin{lemma} \label{lemma:Bound_eigenvalues}
Let $\bm{L}\in\mathbb{R}^{n\times n}$ and $\bm{S}\in\mathbb{R}^{n\times n}$ be two symmetric matrices satisfying Assumptions \nameref{assume:Inco} and \nameref{assume:Sparse}, respectively. Let $\widetilde{\bm{L}}_k\in\mathbb{R}^{n\times n}$  be the trim output of $\bm{L}_k$.  
If \[
\|\bm{L}-\bm{L}_k\|_2 \leq 8\alpha \mu r \gamma^k\sigma_1^L,\quad
\|\bm{S}-\bm{S}_k\|_\infty \leq \frac{\mu r}{n} \gamma^k\sigma_1^L,\textnormal{\ and }
supp(\bm{S}_k)\subset \Omega,
\]
then
\begin{equation}  \label{eq:Bound_eigenvalues 1}
|\sigma^L_i-|\lambda^{(k)}_i|| \leq \tau\sigma_r^L
\end{equation}
and
\begin{equation}  \label{eq:Bound_eigenvalues 2}
(1-2\tau)\gamma^j \sigma^L_1\leq |\lambda^{(k)}_{r+1}| +\gamma^j|\lambda^{(k)}_1|\leq(1+2\tau)\gamma^j \sigma^L_1
\end{equation}
hold for all $k\geq 0$ and $j \leq k+1$,  provided $1>\gamma\geq512\tau r \kappa^2+\frac{1}{\sqrt{12}}$. Here $|\lambda^{(k)}_{i}|$ is the $i^{th}$ singular value of $\mathcal{P}_{\widetilde{T}_k}(\bm{D}-\bm{S}_k)$.
\end{lemma}
\begin{proof}
Since $\bm{D}=\bm{L}+\bm{S}$,  we have
\begin{align*}
\mathcal{P}_{\widetilde{T}_k}(\bm{D}-\bm{S}_k)&=\mathcal{P}_{\widetilde{T}_k}(\bm{L}+\bm{S}-\bm{S}_k)  \cr
&=\bm{L}+(\mathcal{P}_{\widetilde{T}_k}-\mathcal{I})\bm{L}+\mathcal{P}_{\widetilde{T}_k}(\bm{S}-\bm{S}_k).
\end{align*}
Hence, by Weyl's inequality and (\ref{eq:norm_of_Z 1}) in Lemma \ref{lemma:norm_of_Z}, we can see that
\begin{align*}
|\sigma^L_i-|\lambda^{(k)}_i|| &\leq \|(\mathcal{P}_{\widetilde{T}_k}-\mathcal{I})\bm{L}+\mathcal{P}_{\widetilde{T}_k}(\bm{S}-\bm{S}_k)\|_2 \cr
                               &\leq \tau\gamma^{k+1} \sigma_r^L  \KW{\alpha \lesssim\frac{1}{\mu r\kappa}} 
\end{align*}
hold for all $i$ and $k\geq0$.
So the first claim is proved since $\gamma<1$.

Notice that $\bm{L}$ is a rank-$r$ matrix, which implies $\sigma_{r+1}^L=0$, so we have
\[
\begin{split}
||\lambda^{(k)}_{r+1}| +\gamma^{j}|\lambda^{(k)}_1| -\gamma^{j}\sigma_1^L|
						&= ||\lambda^{(k)}_{r+1}| -\sigma_{r+1}^L+\gamma^{j}|\lambda^{(k)}_1|  -\gamma^{j}\sigma_1^L| \cr
						&\leq \tau \gamma^{k+1}\sigma_r^L + \tau \gamma^{j+k+1}\sigma_r^L\cr
						&\leq \left(1+\gamma^{k+1}\right) \tau \gamma^j\sigma_r^L \cr
						&\leq 2\tau \gamma^j\sigma_1^L 
\end{split}
\]
for all $j\leq k+1$. This completes the proof of  the second claim.
\end{proof}

\begin{lemma} \label{lemma:Bound_of_L-L_k_2_norm}
Let $\bm{L}\in\mathbb{R}^{n\times n}$ and $\bm{S}\in\mathbb{R}^{n\times n}$ be two symmetric matrices satisfying Assumptions \nameref{assume:Inco} and \nameref{assume:Sparse}, respectively. Let $\widetilde{\bm{L}}_k\in\mathbb{R}^{n\times n}$  be the trim output of $\bm{L}_k$. 
 If
\[
\|\bm{L}-\bm{L}_k\|_2 \leq 8\alpha \mu r \gamma^k\sigma_1^L,\quad
\|\bm{S}-\bm{S}_k\|_\infty \leq \frac{\mu r}{n} \gamma^k\sigma_1^L,\textnormal{\ and }
supp(\bm{S}_k)\subset \Omega,
\]
then we have 
\[
\|\bm{L}-\bm{L}_{k+1}\|_2 \leq 8\alpha \mu r \gamma^{k+1}\sigma_1^L,
\]
provided $1>\gamma\geq512\tau r \kappa^2+\frac{1}{\sqrt{12}}$.
\end{lemma}
\begin{proof} A direct calculation yields 
\begin{align*}
\|\bm{L}-\bm{L}_{k+1}\|_2 &\leq\|\bm{L}-\mathcal{P}_{\widetilde{T}_k}(\bm{D}-\bm{S}_k)\|_2+\|\mathcal{P}_{\widetilde{T}_k}(\bm{D}-\bm{S}_k)-\bm{L}_{k+1}\|_2\cr
                          &\leq 2\|\bm{L}-\mathcal{P}_{\widetilde{T}_k}(\bm{D}-\bm{S}_k)\|_2\cr
                          &=2\|\bm{L}-\mathcal{P}_{\widetilde{T}_k}(\bm{L}+\bm{S}-\bm{S}_k)\|_2\cr
                          &= 2\|(\mathcal{P}_{\widetilde{T}_k}-\mathcal{I})\bm{L}+\mathcal{P}_{\widetilde{T}_k}(\bm{S}-\bm{S}_k)\|_2 \cr
                          &\leq 2 \cdot \tau\gamma^{k+1}\sigma_r^L \cr
                          &= 8\alpha \mu r\gamma^{k+1}\sigma_1^L,
\end{align*}
where the second inequality follows from the fact $\bm{L}_{k+1}=\mathcal{H}_r(\mathcal{P}_{{\widetilde{T}_k}}(\bm{D}-\bm{S}_k))$ is the best rank-$r$ approximation of $\mathcal{P}_{{\widetilde{T}_k}}(\bm{D}-\bm{S}_k)$, and the last inequality uses (\ref{eq:norm_of_Z 1}) in Lemma \ref{lemma:norm_of_Z}. 
\end{proof}

\begin{lemma}  \label{lemma:Bound_of_L-L_k_infty_norm}
Let $\bm{L}\in\mathbb{R}^{n\times n}$ and $\bm{S}\in\mathbb{R}^{n\times n}$ be two symmetric matrices satisfying Assumptions \nameref{assume:Inco} and \nameref{assume:Sparse}, respectively. Let $\widetilde{\bm{L}}_k\in\mathbb{R}^{n\times n}$  be the trim output of $\bm{L}_k$. If\[
\|\bm{L}-\bm{L}_k\|_2 \leq 8\alpha \mu r \gamma^k\sigma_1^L,\quad
\|\bm{S}-\bm{S}_k\|_\infty \leq \frac{\mu r}{n} \gamma^k\sigma_1^L,\textnormal{\ and }
supp(\bm{S}_k)\subset \Omega,
\]
then we have 
\[
\|\bm{L}-\bm{L}_{k+1}\|_\infty \leq \left(\frac{1}{2}-\tau\right)\frac{\mu r}{n}\gamma^{k+1} \sigma_1^L,
\]
provided $1>\gamma\geq\max\{512\tau r \kappa^2+\frac{1}{\sqrt{12}},\frac{2\upsilon}{(1-12\tau)(1-\tau-\upsilon)^2}\}$ and $\tau<\frac{1}{12}$.
\end{lemma}
\begin{proof}
Let $\mathcal{P}_{\widetilde{T}_k}(\bm{D}-\bm{S}_k)=\begin{bmatrix}\bm{U}_{k+1}& \ddot{\bm{U}}_{k+1} \end{bmatrix} \begin{bmatrix}\bm{\Lambda}  & \bm{0}\\ \bm{0} &\ddot{\bm{\Lambda}}\end{bmatrix} \begin{bmatrix}\bm{U}_{k+1}^T \\ \ddot{\bm{U}}_{k+1}^T\end{bmatrix} =\bm{U}_{k+1}\bm{\Lambda} \bm{U}_{k+1}^T+\ddot{\bm{U}}_{k+1}\ddot{\bm{\Lambda}}\ddot{\bm{U}}_{k+1}^T$ be its eigenvalue decomposition. We use the lighter notation  $\lambda_i$ ($1\leq i\leq n$) for the eigenvalues of $\mathcal{P}_{\widetilde{T}_k}(\bm{D}-\bm{S}_k)$ at the $k$-th iteration and assume they are ordered by $|\lambda_1|\geq|\lambda_2|\geq\cdots\geq|\lambda_n|$. Moreover, $\bm{\Lambda} $ has its $r$ largest eigenvalues in magnitude, $\bm{U}_{k+1}$ contains the first $r$ eigenvectors, and $\ddot{\bm{U}}_{k+1}$ has the rest. It follows that  $ \bm{L}_{k+1}=\mathcal{H}_r(\mathcal{P}_{\widetilde{T}_k}(\bm{D}-\bm{S}_k))=\bm{U}_{k+1}\bm{\Lambda} \bm{U}_{k+1}^T $.

Denote $\bm{Z}=\mathcal{P}_{\widetilde{T}_k}(\bm{D}-\bm{S}_k)-\bm{L}=(\mathcal{P}_{\widetilde{T}_k}-\mathcal{I})\bm{L}+\mathcal{P}_{\widetilde{T}_k}(\bm{S}-\bm{S}_k)$. Let $\bm{u}_i$ be the $i^{th}$ eigenvector of $\mathcal{P}_{\widetilde{T}_k}(\bm{D}-\bm{S}_k)$.  Noting that $(\lambda_i\bm{I}-\bm{Z})\bm{u}_i=\bm{L}\bm{u}_i$, we have 
\begin{align*}
\bm{u}_i = \lb\bm{I}-\frac{\bm{Z}}{\lambda_i}\rb^{-1}  \frac{\bm{L}}{\lambda_i}\bm{u}_i=  \left(\bm{I}+\frac{\bm{Z}}{\lambda_i}+\left(\frac{\bm{Z}}{\lambda_i}\right)^2+\cdots\right)\frac{\bm{L}}{\lambda_i}\bm{u}_i
\end{align*}
for all $\bm{u}_i$ with $1\leq i\leq r$, where the expansion is valid because 
$$\frac{\ln\bm{Z}\rn_2}{\lambda_i}\leq\frac{\ln\bm{Z}\rn_2}{\lambda_r}\leq\frac{\tau}{1-\tau}<1 $$
following from (\ref{eq:norm_of_Z 1}) in Lemma \ref{lemma:norm_of_Z} and (\ref{eq:Bound_eigenvalues 1}) in Lemma \ref{lemma:Bound_eigenvalues}.
 This implies
\begin{align*}
\bm{U}_{k+1}\bm{\Lambda} \bm{U}_{k+1}^T &= \sum_{i=1}^r \bm{u}_i\lambda_i\bm{u}_i^T   \cr
										&= \sum_{i=1}^r \left( \sum_{a\geq0} \left(\frac{\bm{Z}}{\lambda_i}\right)^a\frac{\bm{L}}{\lambda_i} \right)\bm{u}_i\lambda_i\bm{u}_i^T \left( \sum_{b\geq0} \left(\frac{\bm{Z}}{\lambda_i}\right)^b\frac{\bm{L}}{\lambda_i} \right)^T \cr
										&= \sum_{a\geq0}  \bm{Z}^a \bm{L} \sum_{i=1}^r\left(\bm{u}_i\frac{1}{\lambda_i^{a+b+1}}\bm{u}_i^T \right) \bm{L} \sum_{b\geq0} \bm{Z}^b  \cr
										&= \sum_{a,b\geq0}  \bm{Z}^a\bm{L}\bm{U}_{k+1}\bm{\Lambda} ^{-(a+b+1)}\bm{U}_{k+1}^T\bm{L}\bm{Z}^b.
\end{align*}
Thus, we have
\begin{align*}
\|\bm{L}_{k+1}-\bm{L}\|_\infty &= \|\bm{U}_{k+1}\bm{\Lambda} \bm{U}_{k+1}^T -\bm{L}\|_\infty \cr
		             &= \|\bm{L}\bm{U}_{k+1}\bm{\Lambda} ^{-1}\bm{U}_{k+1}^T\bm{L}-\bm{L}  + \sum_{a+b>0} \bm{Z}^a\bm{L}\bm{U}_{k+1}\bm{\Lambda} ^{-(a+b+1)}\bm{U}_{k+1}^T\bm{L}\bm{Z}^b \|_\infty  \cr
		             &\leq \|\bm{L}\bm{U}_{k+1}\bm{\Lambda} ^{-1}\bm{U}_{k+1}^T\bm{L}-\bm{L}\|_\infty  + \sum_{a+b>0} \|\bm{Z}^a\bm{L}\bm{U}_{k+1}\bm{\Lambda} ^{-(a+b+1)}\bm{U}_{k+1}^T\bm{L}\bm{Z}^b \|_\infty \cr
		             &:= \bm{Y}_0 + \sum_{a+b>0} \bm{Y}_{ab}.
\end{align*}

We will handle $\bm{Y}_0$ first. Recall that $\bm{L}=\bm{U}\bm{\Sigma} \bm{V}^T$ is the SVD of the symmetric matrix $\bm{L}$ which obeys $\mu$-incoherence, i.e., $\bm{U}\bm{U}^T=\bm{V}\bm{V}^T$ and $\|\bm{e}_i^T\bm{U}\bm{U}^T\|_2\leq\sqrt{\frac{\mu r}{n}}$ for all $i$. So, for each $(i,j)$ entry of $\bm{Y}_0$, one has 
\begin{align*}
\bm{Y}_0 &= \max_{ij} |\bm{e}_i^T(\bm{L}\bm{U}_{k+1}\bm{\Lambda} ^{-1}\bm{U}_{k+1}^T\bm{L}-\bm{L})\bm{e}_j| \cr
												&=\max_{ij} |\bm{e}_i^T\bm{U}\bm{U}^T(\bm{L}\bm{U}_{k+1}\bm{\Lambda} ^{-1}\bm{U}_{k+1}^T\bm{L}-\bm{L})\bm{U}\bm{U}^T\bm{e}_j| \cr
												&\leq\max_{ij} \|\bm{e}_i^T\bm{U}\bm{U}^T\|_2~\|\bm{L}\bm{U}_{k+1}\bm{\Lambda} ^{-1}\bm{U}_{k+1}^T\bm{L}-\bm{L}\|_2~ \|\bm{U}\bm{U}^T\bm{e}_j\|_2 \cr
												&\leq \frac{\mu r}{n}\|\bm{L}\bm{U}_{k+1}\bm{\Lambda} ^{-1}\bm{U}_{k+1}^T\bm{L}-\bm{L}\|_2,
\end{align*}
where the second equation follows from the fact $\bm{U}\bm{U}^T\bm{L}=\bm{L}\bm{U}\bm{U}^T=\bm{L}$. Since $\bm{L}=\bm{U}_{k+1}\bm{\Lambda} \bm{U}_{k+1}^T+\ddot{\bm{U}}_{k+1}\ddot{\bm{\Lambda}}\ddot{\bm{U}}_{k+1}^T -\bm{Z}$, there hold
\begin{align*}
&~\|\bm{L}\bm{U}_{k+1}\bm{\Lambda} ^{-1}\bm{U}_{k+1}^T\bm{L}-\bm{L}\|_2 \cr
=&~ \|(\bm{U}_{k+1}\bm{\Lambda} \bm{U}_{k+1}^T+\ddot{\bm{U}}_{k+1}\ddot{\bm{\Lambda}}\ddot{\bm{U}}_{k+1}^T -\bm{Z})\bm{U}_{k+1}\bm{\Lambda} ^{-1}\bm{U}_{k+1}^T(\bm{U}_{k+1}\bm{\Lambda} \bm{U}_{k+1}^T+\ddot{\bm{U}}_{k+1}\ddot{\bm{\Lambda}}\ddot{\bm{U}}_{k+1}^T -\bm{Z})-\bm{L}\|_2  \cr
=&~\|\bm{U}_{k+1}\bm{\Lambda} \bm{U}_{k+1}^T-\bm{L}-\bm{U}_{k+1}\bm{U}_{k+1}^T\bm{Z}-\bm{Z}\bm{U}_{k+1}\bm{U}_{k+1}^T-\bm{Z}\bm{U}_{k+1}\bm{\Lambda} ^{-1}\bm{U}_{k+1}^T\bm{Z}\|_2  \cr
\leq&~ \|\bm{Z}-\ddot{\bm{U}}_{k+1}\ddot{\bm{\Lambda}}\ddot{\bm{U}}_{k+1}^T\|_2 +2\|\bm{Z}\|_2+\frac{\|\bm{Z}\|_2^2}{|\lambda_r|}  \cr
\leq&~ \|\ddot{\bm{U}}_{k+1}\ddot{\bm{\Lambda}}\ddot{\bm{U}}_{k+1}^T\|_2 +4\|\bm{Z}\|_2  \cr
\leq&~  |\lambda_{r+1}|+4\|\bm{Z}\|_2 \cr
\leq&~  5\|\bm{Z}\|_2  \cr
\leq&~  5\tau\gamma^{k+1}\sigma_1^L, \KW{\alpha \lesssim\frac{1}{\mu r}}
\end{align*}
where the fifth inequality uses (\ref{eq:norm_of_Z 1}) in Lemma \ref{lemma:norm_of_Z}, 
and notice that $\frac{\|\bm{Z}\|_2}{|\lambda_r|}\leq\frac{\tau}{1-\tau}<1$ since $\tau<\frac{1}{2}$ and  $|\lambda_{r+1}|\leq\|\bm{Z}\|_2$ since $\bm{L}$ is a rank-$r$ matrix. 
Thus, we have
\begin{equation}  \label{eq:Y0 bound}
\bm{Y}_0\leq \frac{\mu r}{n} 5\tau\gamma^{k+1}\sigma_1^L.
\end{equation}

Next, we derive an upper bound for the rest part. Note that
\begin{align*}
\bm{Y}_{ab}&=\max_{ij} |\bm{e}_i^T\bm{Z}^a\bm{L}\bm{U}_{k+1}\bm{\Lambda} ^{-(a+b+1)}\bm{U}_{k+1}^T\bm{L}\bm{Z}^b\bm{e}_j|  \cr
      &=\max_{ij} |(\bm{e}_i^T\bm{Z}^a\bm{U}\bm{U}^T)\bm{L}\bm{U}_{k+1}\bm{\Lambda} ^{-(a+b+1)}\bm{U}_{k+1}^T\bm{L}(\bm{U}\bm{U}^T\bm{Z}^b\bm{e}_j)|  \cr
      &\leq \max_{ij} \|\bm{e}_i^T\bm{Z}^a\bm{U}\|_2~\|\bm{L}\bm{U}_{k+1}\bm{\Lambda} ^{-(a+b+1)}\bm{U}_{k+1}^T\bm{L}\|_2~\|\bm{U}^T\bm{Z}^b\bm{e}_j\|_2  \cr
      &\leq \max_l\frac{\mu r}{n}( \sqrt{n}\|\bm{e}_l^T \bm{Z}\|_2)^{a+b} \|\bm{L}\bm{U}_{k+1}\bm{\Lambda} ^{-(a+b+1)}\bm{U}_{k+1}^T\bm{L}\|_2,
\end{align*}
where the last inequality uses Lemma \ref{lemma:bound_power_vector_norm_with_incoherence}.
Furthermore, by using $\bm{L}=\bm{U}_{k+1}\bm{\Lambda} \bm{U}_{k+1}^T+\ddot{\bm{U}}_{k+1}\ddot{\bm{\Lambda}}\ddot{\bm{U}}_{k+1}^T -\bm{Z}$ again, we get
\begin{align*}
&~\|\bm{L}\bm{U}_{k+1}\bm{\Lambda} ^{-(a+b+1)}\bm{U}_{k+1}^T\bm{L}\|_2 \cr
=&~ \|(\bm{U}_{k+1}\bm{\Lambda} \bm{U}_{k+1}^T+\ddot{\bm{U}}_{k+1}\ddot{\bm{\Lambda}}\ddot{\bm{U}}_{k+1}^T -\bm{Z})\bm{U}_{k+1}\bm{\Lambda} ^{-(a+b+1)}\bm{U}_{k+1}^T(\bm{U}_{k+1}\bm{\Lambda} \bm{U}_{k+1}^T+\ddot{\bm{U}}_{k+1}\ddot{\bm{\Lambda}}\ddot{\bm{U}}_{k+1}^T -\bm{Z})\|_2        \cr
=&~ \|\bm{U}_{k+1}\bm{\Lambda} ^{-(a+b-1)}\bm{U}_{k+1}^T-\bm{Z}\bm{U}_{k+1}\bm{\Lambda} ^{-(a+b)}\bm{U}_{k+1}^T-\bm{U}_{k+1}\bm{\Lambda} ^{-(a+b)}\bm{U}_{k+1}^T\bm{Z}+\bm{Z}\bm{U}_{k+1}\bm{\Lambda} ^{-(a+b+1)}\bm{U}_{k+1}^T\bm{Z}\|_2 \cr
\leq&~ |\lambda_r|^{-(a+b-1)} + |\lambda_r|^{-(a+b)}\|\bm{Z}\|_2+ |\lambda_r|^{-(a+b)}\|\bm{Z}\|_2+ |\lambda_r|^{-(a+b+1)}\|\bm{Z}\|_2^2 \cr
=&~ |\lambda_r|^{-(a+b-1)}\left( 1+ \frac{2\|\bm{Z}\|_2}{|\lambda_r|}+\left(\frac{\|\bm{Z}\|_2}{|\lambda_r|}\right)^2 \right)  \cr
=&~ |\lambda_r|^{-(a+b-1)}\left( 1+ \frac{\|\bm{Z}\|_2}{|\lambda_r|} \right)^2  \cr
\leq&~ |\lambda_r|^{-(a+b-1)}\left(  \frac{1}{1-\tau} \right)^2  \cr
\leq&~ \left( \frac{1}{1-\tau} \right)^2 \left((1-\tau)\sigma_r^L\right)^{-(a+b-1)},
\end{align*}
where the second inequality follows from $\frac{\|\bm{Z}\|_2}{|\lambda_r|}\leq\frac{\tau}{1-\tau}$, and the last inequality follows from Lemma \ref{lemma:Bound_eigenvalues}. Together with (\ref{eq:norm_of_Z 2}) in Lemma \ref{lemma:norm_of_Z}, we have
\begin{align*}
\sum_{a+b>0}\bm{Y}_{ab}
	&\leq  \sum_{a+b>0}\frac{\mu r}{n}\left( \frac{1}{1-\tau} \right)^2 \upsilon\gamma^{k}\sigma_r^L \left( \frac{\upsilon\gamma^{k}\sigma_r^L}{(1-\tau)\sigma_r^L}\right)^{a+b-1}\cr
    &\leq \frac{\mu r}{n}\left( \frac{1}{1-\tau} \right)^2 \upsilon \gamma^{k}\sigma_1^L  \sum_{a+b>0} \left(  \frac{\upsilon}{1-\tau}\right)^{a+b-1}\cr
    &\leq \frac{\mu r}{n}\left( \frac{1}{1-\tau} \right)^2 \upsilon \gamma^{k}\sigma_1^L  \left( \frac{1}{1-\frac{\upsilon}{1-\tau}}\right)^2\cr
    &\leq \frac{\mu r}{n}\left( \frac{1}{1-\tau-\upsilon} \right)^2 \upsilon \gamma^{k}\sigma_1^L. \numberthis\label{eq:sum Y_ab bound}
\end{align*}
Finally, combining (\ref{eq:Y0 bound}) and (\ref{eq:sum Y_ab bound}) together gives
\begin{align*}
\|\bm{L}_{k+1}-\bm{L}\|_\infty &= \bm{Y}_0 + \sum_{a+b>0} \bm{Y}_{ab}  \cr
                     &\leq \frac{\mu r}{n} 5\tau\gamma^{k+1}\sigma_1^L + \frac{\mu r}{n}\left( \frac{1}{1-\tau-\upsilon} \right)^2 \upsilon \gamma^{k}\sigma_1^L \cr
                     &\leq \left(\frac{1}{2}-\tau\right)\frac{\mu r}{n} \gamma^{k+1}\sigma_1^L,
\end{align*}
where the last inequality follows from $\gamma\geq\frac{2\upsilon}{(1-12\tau)(1-\tau-\upsilon)^2}$.
\end{proof}

\begin{lemma}   \label{lemma:Bound_of_S-S_k}
Let $\bm{L}\in\mathbb{R}^{n\times n}$ and $\bm{S}\in\mathbb{R}^{n\times n}$ be two symmetric matrices satisfying Assumptions \nameref{assume:Inco} and \nameref{assume:Sparse}, respectively. Let $\widetilde{\bm{L}}_k\in\mathbb{R}^{n\times n}$  be the trim output of $\bm{L}_k$. Recall that  $\beta=\frac{\mu r}{2n}$. If 
\[
\|\bm{L}-\bm{L}_k\|_2 \leq 8\alpha \mu r \gamma^k\sigma_1^L,\quad
\|\bm{S}-\bm{S}_k\|_\infty \leq \frac{\mu r}{n} \gamma^k\sigma_1^L,\textnormal{\ and }
supp(\bm{S}_k)\subset \Omega
\]
 then we have
\[
supp(\bm{S}_{k+1})\subset \Omega \quad and \quad \|\bm{S}-\bm{S}_{k+1}\|_\infty \leq  \frac{\mu r}{n} \gamma^{k+1} \sigma_1^L,
\]
provided $1>\gamma\geq\max\{512\tau r \kappa^2+\frac{1}{\sqrt{12}},\frac{2\upsilon}{(1-12\tau)(1-\tau-\upsilon)^2}\}$ and $\tau<\frac{1}{12}$.
\end{lemma}
\begin{proof}
We first notice that
\[
[\bm{S}_{k+1}]_{ij}=[\mathcal{T}_{\zeta_{k+1}} (\bm{D}-\bm{L}_{k+1})]_{ij}=[\mathcal{T}_{\zeta_{k+1}} (\bm{S}+\bm{L}-\bm{L}_{k+1})]_{ij}=
\begin{cases}
\mathcal{T}_{\zeta_{k+1}} ([\bm{S}+\bm{L}-\bm{L}_{k+1}]_{ij}) & (i,j)\in\Omega \cr
\mathcal{T}_{\zeta_{k+1}} ([\bm{L}-\bm{L}_{k+1}]_{ij})        & (i,j)\in\Omega^c  \cr
\end{cases}.
\]
Let $|\lambda_{i}^{(k)}|$ denote $i^{th}$ singular value of $\mathcal{P}_{\widetilde{T}_k}(\bm{D}-\bm{S}_k)$. By Lemmas \ref{lemma:Bound_eigenvalues} and \ref{lemma:Bound_of_L-L_k_infty_norm}, we have
\[
\begin{split}
|[\bm{L}-\bm{L}_{k+1}]_{ij}|&\leq\|\bm{L}-\bm{L}_{k+1}\|_\infty\leq  \left(\frac{1}{2}-\tau\right)\frac{\mu r}{n} \gamma^{k+1} \sigma_1^L \\
&\leq \left(\frac{1}{2}-\tau\right)\frac{\mu r}{n} \frac{1}{1-2\tau}\left(|\lambda_{r+1}^{(k)}| +\gamma^{k+1}|\lambda_1^{(k)}|\right) \\
&= \zeta_{k+1}
\end{split}
\]
for any entry of $\bm{L}-\bm{L}_{k+1}$.
Hence, $[\bm{S}_{k+1}]_{ij}=0$ for all $(i,j)\in \Omega^c$, i.e., $supp(\bm{S}_{k+1})\subset \Omega$. 

Denote $\Omega_{k+1}:=supp(\bm{S}_{k+1})=\{(i,j)~|~[(\bm{D}-\bm{L}_{k+1})]_{ij}>\zeta_k\}$. Then, for any entry of $\bm{S}-\bm{S}_{k+1}$, there hold
\[
[\bm{S}-\bm{S}_{k+1}]_{ij} =
\begin{cases}
0 &\cr
[\bm{L}_{k+1}-\bm{L}]_{ij} &   \cr
[\bm{S}]_{ij}         & \cr
\end{cases} \leq
\begin{cases}
0  &\cr
\|\bm{L}-\bm{L}_{k+1}\|_\infty       &  \cr
\|\bm{L}-\bm{L}_{k+1}\|_\infty +\zeta_{k+1} & \cr
\end{cases} \leq
\begin{cases}
0 & (i,j)\in\Omega^c  \cr
\left(\frac{1}{2}-\tau\right)\frac{\mu r}{n} \gamma^{k+1} \sigma_1^L       & (i,j)\in \Omega_{k+1}   \cr
\frac{\mu r}{n} \gamma^{k+1} \sigma_1^L & (i,j)\in \Omega\backslash\Omega_{k+1}. \cr
\end{cases}
\]
Here the last step follows from Lemma \ref{lemma:Bound_eigenvalues} which implies $\zeta_{k+1}=\frac{\mu r}{2n} (|\lambda_{r+1}^{(k)}| +\gamma^{k+1}|\lambda_1^{(k)}|)\leq \left(\frac{1}{2}+\tau\right)\frac{\mu r}{n}\gamma^{k+1}\sigma_1^L$.
Therefore, $\|\bm{S}-\bm{S}_{k+1}\|_\infty\leq \frac{\mu r}{n} \gamma^{k+1} \sigma_1^L$. 
\end{proof}

Now, we have all the ingredients for the proof of Theorem \ref{thm:local convergence}.
\begin{proof} [Proof of Theorem \ref{thm:local convergence}] This theorem will be proved by  mathematical induction.\\
\textbf{Base Case:} When $k=0$, the base case is satisfied by the assumption on the intialization.\\
\textbf{Induction Step:} Assume we have
\[
\|\bm{L}-\bm{L}_k\|_2 \leq 8\alpha \mu r \gamma^k\sigma_1^L,\quad
\|\bm{S}-\bm{S}_k\|_\infty \leq \frac{\mu r}{n} \gamma^k\sigma_1^L,\quad \textnormal{and} \quad
supp(\bm{S}_k)\subset \Omega
\]
at the $k^{th}$ iteration. At the $(k+1)^{th}$ iteration. If follows directly from Lemmas~\ref{lemma:Bound_of_L-L_k_2_norm} and \ref{lemma:Bound_of_S-S_k}  that
\[
\|\bm{L}-\bm{L}_{k+1}\|_2 \leq 8\alpha \mu r \gamma^{k+1}\sigma_1^L,\quad \|\bm{S}-\bm{S}_{k+1}\|_\infty \leq \frac{\mu r}{n} \gamma^{k+1}\sigma_1^L\quad \textnormal{and} \quad
supp(\bm{S}_{k+1})\subset \Omega,
\]
which completes the proof.

Additionally, notice that we overall require $1>\gamma\geq\max\{512\tau r \kappa^2+\frac{1}{\sqrt{12}},\frac{2\upsilon}{(1-12\tau)(1-\tau-\upsilon)^2}\}$. By the definition of $\tau$ and $\upsilon$, one can easily see that the lower bound approaches $\frac{1}{\sqrt{12}}$ when the constant hidden in \eqref{eq:condition_on_p} is sufficiently large. Therefore,  the theorem can be proved for any $\gamma\in\left(\frac{1}{\sqrt{12}},1\right)$.
\end{proof}


\subsection{Proof of Theorem \ref{thm:initialization bound}} \label{subsec:proof of initialization}
We first present a lemma which is a variant of Lemma~\ref{lemma:bound_power_vector_norm_with_incoherence} and also appears in \citep[Lemma 5]{netrapalli2014non}. The lemma can be similarly proved by mathematical induction. 
\begin{lemma} \label{init:lemma:bound_power_vector_norm_with_incoherence}
Let $\bm{S}\in\mathbb{R}^{n\times n}$ be a sparse matrix satisfying Assumption \nameref{assume:Sparse}. Let $\bm{U}\in\mathbb{R}^{n\times r}$ be an orthogonal matrix with $\mu$-incoherence, i.e., $\|\bm{e}_i^T\bm{U}\|_2\leq\sqrt{\frac{\mu r}{n}}$ for all $i$. Then
\[
\|\bm{e}_i^T\bm{S}^a\bm{U}\|_2\leq \sqrt{\frac{\mu r}{n}}(\alpha n\|\bm{S}\|_\infty)^a
\]
for all $i$ and $a\geq 0$.
\end{lemma}


Though the proposed initialization scheme (i.e., Algorithms~\ref{Algo:Init1}) basically consists of two steps of AltProj \citep{netrapalli2014non}, we provide an independent proof for Theorem~\ref{thm:initialization bound} here because we bound the approximation errors of the low rank matrices using the spectral norm rather than the infinity norm. The proof of Theorem~\ref{thm:initialization bound} follows a similar structure to that of Theorem~\ref{thm:local convergence}, but without the projection onto a low dimensional tangent space. Instead of 
first presenting several auxiliary lemmas, we give a single proof by putting all the elements together. 

\begin{proof} [Proof of Theorem \ref{thm:initialization bound}] The proof can be partitioned into several parts.

~\\
(i) Note that $\bm{L}_{-1}=0$ and
\[
\|\bm{L}-\bm{L}_{-1}\|_\infty=\|\bm{L}\|_\infty = \max_{ij} \lab\bm{e}_i^T\bm{U}\bm{\Sigma} \bm{U}^T\bm{e}_j\rab \leq \max_{ij} \|\bm{e}_i^T\bm{U}\|_2\|\bm{\Sigma}\|_2\|\bm{U}^T\bm{e}_j\|_2\leq\frac{\mu r}{n}\sigma_1^L,
\]
where the last inequality follows from Assumption \nameref{assume:Inco}, i.e., $\bm{L}$ is $\mu$-incoherent.
Thus, with the choice of $\beta_{init}\geq\frac{\mu r\sigma_1^L}{n\sigma_1^D}$, we have
\begin{equation} \label{eq:Init:step 1 result}
\|\bm{L}-\bm{L}_{-1}\|_\infty\leq \beta_{init}\sigma_1^D = \zeta_{-1}.
\end{equation}
Since
\[
[\bm{S}_{-1}]_{ij}=[\mathcal{T}_{\zeta_{-1}} (\bm{S}+\bm{L}-\bm{L}_{-1})]_{ij} =
\begin{cases}
\mathcal{T}_{\zeta_{-1}} ([\bm{S}+\bm{L}-\bm{L}_{-1}]_{ij}) & (i,j)\in\Omega \cr
\mathcal{T}_{\zeta_{-1}} ([\bm{L}-\bm{L}_{-1}]_{ij})        & (i,j)\in\Omega^c,  \cr
\end{cases}
\]
it follows that  $[\bm{S}_{-1}]_{ij}=0$ for all $(i,j)\in\Omega^c$, i.e. $\Omega_{-1}:=supp(\bm{S}_{-1})\subset \Omega$. Moreover, for any entries of $\bm{S}-\bm{S}_{-1}$, we have
\[
[\bm{S}-\bm{S}_{-1}]_{ij} =
\begin{cases}
0 &  \cr
[\bm{L}_{-1}-\bm{L}]_{ij} &   \cr
[\bm{S}]_{ij}         & \cr
\end{cases} \leq
\begin{cases}
0 &  \cr
\|\bm{L}-\bm{L}_{-1}\|_\infty       &  \cr
\|\bm{L}-\bm{L}_{-1}\|_\infty +\zeta_{-1} & \cr
\end{cases} \leq
\begin{cases}
0 &  (i,j)\in \Omega^c   \cr
\frac{\mu r}{n} \sigma_1^L       & (i,j)\in \Omega_{-1}   \cr
\frac{4\mu r}{n}  \sigma_1^L & (i,j)\in \Omega\backslash\Omega_{-1} \cr
\end{cases},
\]
where the last inequality follows from $\beta_{init}\leq\frac{3\mu r\sigma_1^L}{n\sigma_1^D}$, so that $\zeta_{-1}\leq \frac{3\mu r}{n}\sigma_1^L$. Therefore, if follows that
\begin{equation}\label{eq:SminusS}
supp(\bm{S}_{-1})\subset\Omega\quad \textnormal{and}\quad \|\bm{S}-\bm{S}_{-1}\|_\infty\leq\frac{4\mu r}{n}\sigma_1^L.
\end{equation}
By Lemma \ref{lemma:bound of sparse matrix}, we also have
\[
\|\bm{S}-\bm{S}_{-1}\|_2\leq \alpha n\|\bm{S}-\bm{S}_{-1}\|_\infty\leq4\alpha \mu r\sigma_1^L.
\]

~\\
(ii) To bound the approximation error of $\bm{L}_0$  to $\bm{L}$ in terms of the spectral norm, note that
\[
\begin{split}
\|\bm{L}-\bm{L}_0\|_2&\leq\|\bm{L}-(\bm{D}-\bm{S}_{-1})\|_2 + \|(\bm{D}-\bm{S}_{-1})-\bm{L}_0\|_2 \cr
           &\leq 2\|\bm{L}-(\bm{D}-\bm{S}_{-1})\|_2\cr
           &= 2\|\bm{L}-(\bm{L}+\bm{S}-\bm{S}_{-1})\|_2\cr
           &= 2\|\bm{S}-\bm{S}_{-1}\|_2,
\end{split}
\]
where the second inequality follows from the fact $\bm{L}_{0}=\mathcal{H}_r(\bm{D}-\bm{S}_{-1})$ is the best rank-$r$ approximation of $\bm{D}-\bm{S}_{-1}$. It follows immediately that
\begin{equation}\label{eq:norm:L-L0}
\|\bm{L}-\bm{L}_0\|_2\leq 8\alpha \mu r\sigma_1^L  .
\end{equation}

~\\
(iii) Since $\bm{D}=\bm{L}+\bm{S}$, we have $\bm{D}-\bm{S}_{-1}=\bm{L}+\bm{S}-\bm{S}_{-1}$. Let $\lambda_i$ denotes the $i^{th}$ eigenvalue of $\bm{D}-\bm{S}_{-1}$ ordered by $|\lambda_1|\geq|\lambda_2|\geq\cdots\geq|\lambda_n|$. The application of  Weyl's inequality together with the bound of $\alpha $ in Assumption \nameref{assume:Sparse} implies that
\begin{equation}   \label{init:eq:eigenvalues bound 0}
|\sigma^L_i-|\lambda_i|| \leq \|\bm{S}-\bm{S}_{-1}\|_2 \leq \frac{\sigma_r^L}{8} \KW{\alpha \lesssim\frac{1}{\mu r\kappa}}
\end{equation}
holds for all $i$.
Consequently, we have 
\begin{align*} 
&\frac{7}{8}\sigma^L_i \leq |\lambda_i| \leq \frac{9}{8}\sigma^L_i,\qquad \forall 1\leq i\leq r,\numberthis  \label{init:eq:eigenvalues bound 1}\\
&\frac{\|\bm{S}-\bm{S}_{-1}\|_2}{\lab\lambda_r\rab}\leq \frac{\frac{\sigma_r^L}{8}}{\frac{7\sigma_r^L}{8}}=\frac{1}{7}.\numberthis\label{init:eq:eigenvalues bound 2}
\end{align*}

Let $\bm{D}-\bm{S}_{-1}=[\bm{U}_0, \ddot{\bm{U}}_0 ]\begin{bmatrix}\bm{\Lambda}  & \bm{0}\\\bm{0} &\ddot{\bm{\Lambda}}\end{bmatrix}[\bm{U}_0, \ddot{\bm{U}}_0]^T =\bm{U}_0\bm{\Lambda} \bm{U}_0^T+\ddot{\bm{U}}_0\ddot{\bm{\Lambda}}\ddot{\bm{U}}_0^T$ be its eigenvalue decomposition, where $\bm{\Lambda} $ has the $r$ largest eigenvalues in magnitude and $\ddot{\bm{\Lambda}}$ contains the rest eigenvalues. Also, $\bm{U}_0$ contains the first $r$ eigenvectors, and $\ddot{\bm{U}}_0$ has the rest. Notice that  $ \bm{L}_0=\mathcal{H}_r(\bm{D}-\bm{S}_{-1})=\bm{U}_0\bm{\Lambda} \bm{U}_0^T $ due to the symmetric setting.
Denote $\bm{E}=\bm{D}-\bm{S_{-1}}-\bm{L}=\bm{S}-\bm{S}_{-1}$. Let $\bm{u}_i$ be the $i^{th}$ eigenvector of $\bm{D}-\bm{S}_{-1}=\bm{L}+\bm{E}$. For $1\leq i\leq r$, since $(\bm{L}+\bm{E})\bm{u}_i =  \lambda_i \bm{u}_i$, we have
\begin{align*}
\bm{u}_i =   \left(\bm{I}-\frac{\bm{E}}{\lambda_i}\right)^{-1}\frac{\bm{L}}{\lambda_i}\bm{u}_i=\left(\bm{I}+\frac{\bm{E}}{\lambda_i}+\left(\frac{\bm{E}}{\lambda_i}\right)^2+\cdots\right)\frac{\bm{L}}{\lambda_i}\bm{u}_i
\end{align*}
for each $\bm{u}_i$, where the expansion in the last equality is valid because $ \frac{\ln\bm{E}\rn_2}{\lab\lambda_i\rab}\leq\frac{1}{7}$ for all $1\leq i\leq r$ following from \eqref{init:eq:eigenvalues bound 2}. Therefore, $\KW{E \mbox{ SYMMETRIC}}$
\begin{align*}
\|\bm{L}_0-\bm{L}\|_\infty &= \|\bm{U}_0\bm{\Lambda} \bm{U}_0^T -\bm{L}\|_\infty \cr
		             &= \|\bm{L}\bm{U}_0\bm{\Lambda} ^{-1}\bm{U}_0^T\bm{L}-\bm{L}  + \sum_{a+b>0} \bm{E}^a\bm{L}\bm{U}_0\bm{\Lambda} ^{-(a+b+1)}\bm{U}_0^T\bm{L}\bm{E}^b \|_\infty  \cr
		             &\leq \|\bm{L}\bm{U}_0\bm{\Lambda} ^{-1}\bm{U}_0^T\bm{L}-\bm{L}\|_\infty  + \sum_{a+b>0} \|\bm{E}^a\bm{L}\bm{U}_0\bm{\Lambda} ^{-(a+b+1)}\bm{U}_0^T\bm{L}\bm{E}^b \|_\infty \cr
		             &:= \bm{Y}_0 + \sum_{a+b>0} \bm{Y}_{ab}.
\end{align*}

We will handle $\bm{Y}_0$ first. Recall that $\bm{L}=\bm{U}\bm{\Sigma} \bm{V}^T$ is the SVD of the symmetric matrix $\bm{L}$ which is $\mu$-incoherence, i.e., $\bm{U}\bm{U}^T=\bm{V}\bm{V}^T$ and $\|\bm{e}_i^T\bm{U}\bm{U}^T\|_2\leq\sqrt{\frac{\mu r}{n}}$ for all $i$. For each $(i,j)$ entry of $\bm{Y}_0$, we have
\begin{align*}
\bm{Y}_0 &= \max_{ij} |\bm{e}_i^T(\bm{L}\bm{U}_0\bm{\Lambda} ^{-1}\bm{U}_0^T\bm{L}-\bm{L})\bm{e}_j| \cr
												&=\max_{ij} |\bm{e}_i^T\bm{U}\bm{U}^T(\bm{L}\bm{U}_0\bm{\Lambda} ^{-1}\bm{U}_0^T\bm{L}-\bm{L})\bm{U}\bm{U}^T\bm{e}_j| \cr
												&\leq\max_{ij} \|\bm{e}_i^T\bm{U}\bm{U}^T\|_2~\|\bm{L}\bm{U}_0\bm{\Lambda} ^{-1}\bm{U}_0^T\bm{L}-\bm{L}\|_2~ \|\bm{U}\bm{U}^T\bm{e}_j\|_2 \cr		
											&\leq \frac{\mu r}{n}\|\bm{L}\bm{U}_0\bm{\Lambda} ^{-1}\bm{U}_0^T\bm{L}-\bm{L}\|_2,
\end{align*}
where the second equation follows from the fact $\bm{L}=\bm{U}\bm{U}^T\bm{L}=\bm{L}\bm{U}\bm{U}^T$. Since $\bm{L}=\bm{U}_0\bm{\Lambda} \bm{U}_0^T+\ddot{\bm{U}}_0\ddot{\bm{\Lambda}}\ddot{\bm{U}}_0^T -\bm{E}$,
\begin{align*}
&~\|\bm{L}\bm{U}_0\bm{\Lambda} ^{-1}\bm{U}_0^T\bm{L}-\bm{L}\|_2 \cr
=&~ \|(\bm{U}_0\bm{\Lambda} \bm{U}_0^T+\ddot{\bm{U}}_0\ddot{\bm{\Lambda}}\ddot{\bm{U}}_0^T -\bm{E})\bm{U}_0\bm{\Lambda} ^{-1}\bm{U}_0^T(\bm{U}_0\bm{\Lambda} \bm{U}_0^T+\ddot{\bm{U}}_0\ddot{\bm{\Lambda}}\ddot{\bm{U}}_0^T -\bm{E})-\bm{L}\|_2  \cr
=&~\|\bm{U}_0\bm{\Lambda} \bm{U}_0^T-\bm{L}-\bm{U}_0\bm{U}_0^T\bm{E}-\bm{E}\bm{U}_0\bm{U}_0^T-\bm{E}\bm{U}_0\bm{\Lambda} ^{-1}\bm{U}_0^T\bm{E}\|_2  \cr
\leq&~ \|\bm{E}-\ddot{\bm{U}}_0\ddot{\bm{\Lambda}}\ddot{\bm{U}}_0^T\|_2 +2\|\bm{E}\|_2+\frac{\|\bm{E}\|_2^2}{|\lambda_r|}  \cr
\leq&~ \|\ddot{\bm{U}}_0\ddot{\bm{\Lambda}}\ddot{\bm{U}}_0^T\|_2 +4\|\bm{E}\|_2  \cr
\leq&~  |\lambda_{r+1}|+4\|\bm{E}\|_2 \cr
\leq&~  5\|\bm{E}\|_2,
\end{align*}
where the first and fourth inequality follow from (\ref{init:eq:eigenvalues bound 0}) and (\ref{init:eq:eigenvalues bound 2}), and $|\lambda_{r+1}|\leq\|\bm{E}\|_2$ since $\sigma_{r+1}^L=0$. Together, we have
\begin{equation}  \label{init:eq:Y0 bound}
\bm{Y}_0\leq \frac{5\mu r}{n} \|\bm{E}\|_2 \leq 5\alpha \mu r \|\bm{E}\|_\infty,
\end{equation}
where the last inequality follows from Lemma \ref{lemma:bound of sparse matrix}.

Next, we will find an upper bound for the rest part. Note that
\begin{align*}
\bm{Y}_{ab}&=\max_{ij} |\bm{e}_i^T\bm{E}^a\bm{L}\bm{U}_0\bm{\Lambda} ^{-(a+b+1)}\bm{U}_0^T\bm{L}\bm{E}^b\bm{e}_j|  \cr
      &=\max_{ij} |(\bm{e}_i^T\bm{E}^a\bm{U}\bm{U}^T)\bm{L}\bm{U}_0\bm{\Lambda} ^{-(a+b+1)}\bm{U}_0^T\bm{L}(\bm{U}\bm{U}^T\bm{E}^b\bm{e}_j)|  \cr
      &\leq \max_{ij} \|\bm{e}_i^T\bm{E}^a\bm{U}\|_2~\|\bm{L}\bm{U}_0\bm{\Lambda} ^{-(a+b+1)}\bm{U}_0^T\bm{L}\|_2~\|\bm{U}^T\bm{E}^b\bm{e}_j\|_2  \cr
      &\leq \frac{\mu r}{n}( \alpha n\|\bm{E}\|_\infty)^{a+b} \|\bm{L}\bm{U}_0\bm{\Lambda} ^{-(a+b+1)}\bm{U}_0^T\bm{L}\|_2 \cr
      &\leq \alpha \mu r \|\bm{E}\|_\infty \left(\frac{\sigma_r^L}{8}\right)^{a+b-1} \|\bm{L}\bm{U}_0\bm{\Lambda} ^{-(a+b+1)}\bm{U}_0^T\bm{L}\|_2, \KW{ \alpha \lesssim\frac{1}{\mu r\kappa}}
\end{align*}
where the second inequality uses Lemma \ref{init:lemma:bound_power_vector_norm_with_incoherence}. 
Furthermore, by using $\bm{L}=\bm{U}_0\bm{\Lambda} \bm{U}_0^T+\ddot{\bm{U}}_0\ddot{\bm{\Lambda}}\ddot{\bm{U}}_0^T -\bm{E}$ again, we have
\begin{align*}
&~\|\bm{L}\bm{U}_0\bm{\Lambda} ^{-(a+b+1)}\bm{U}_0^T\bm{L}\|_2 \cr
=&~ \|(\bm{U}_0\bm{\Lambda} \bm{U}_0^T+\ddot{\bm{U}}_0\ddot{\bm{\Lambda}}\ddot{\bm{U}}_0^T -\bm{E})\bm{U}_0\bm{\Lambda} ^{-(a+b+1)}\bm{U}_0^T(\bm{U}_0\bm{\Lambda} \bm{U}_0^T+\ddot{\bm{U}}_0\ddot{\bm{\Lambda}}\ddot{\bm{U}}_0^T -\bm{E})\|_2        \cr
=&~ \|\bm{U}_0\bm{\Lambda} ^{-(a+b-1)}\bm{U}_0^T-\bm{E}\bm{L}\bm{U}_0\bm{\Lambda} ^{-(a+b)}\bm{U}_0^T-\bm{L}\bm{U}_0\bm{\Lambda} ^{-(a+b)}\bm{U}_0^T\bm{E}+\bm{E}\bm{L}\bm{U}_0\bm{\Lambda} ^{-(a+b+1)}\bm{U}_0^T\bm{E}\|_2 \cr
\leq&~ |\lambda_r|^{-(a+b-1)} + |\lambda_r|^{-(a+b)}\|\bm{E}\|_2+ |\lambda_r|^{-(a+b)}\|\bm{E}\|_2+ |\lambda_r|^{-(a+b+1)}\|\bm{E}\|_2^2 \cr
=&~ |\lambda_r|^{-(a+b-1)}\left( 1+ \frac{2\|\bm{E}\|_2}{|\lambda_r|}+\left(\frac{\|\bm{E}\|_2}{|\lambda_r|}\right)^2 \right)  \cr
=&~ |\lambda_r|^{-(a+b-1)}\left( 1+ \frac{\|\bm{E}\|_2}{|\lambda_r|} \right)^2  \cr
\leq&~ 2|\lambda_r|^{-(a+b-1)} \cr
\leq&~ 2\left(\frac{7}{8}\sigma_r^L\right)^{-(a+b-1)},
\end{align*}
where the second inequality follows from \eqref{init:eq:eigenvalues bound 2} and the last inequality follows from \eqref{init:eq:eigenvalues bound 1}. Together, we have
\begin{align*}
\sum_{a+b>0}\bm{Y}_{ab}&\leq  \sum_{a+b>0} 2\alpha \mu r \|\bm{E}\|_\infty \left(\frac{\frac{1}{8}\sigma_r^L}{\frac{7}{8}\sigma_r^L }\right)^{a+b-1}  \cr
                  &\leq 2\alpha \mu r \|\bm{E}\|_\infty  \sum_{a+b>0} \left(  \frac{1}{7}\right)^{a+b-1}\cr
                  &\leq 2\alpha \mu r \|\bm{E}\|_\infty  \left(  \frac{1}{1-\frac{1}{7}}\right)^2\cr
                  &\leq 3\alpha \mu r \|\bm{E}\|_\infty.  \numberthis\label{init:eq:sum Y_ab bound}
\end{align*}
Finally, combining \eqref{init:eq:Y0 bound}) and \eqref{init:eq:sum Y_ab bound}) together yields
\begin{align*}
\|\bm{L}_0-\bm{L}\|_\infty &= \bm{Y}_0 + \sum_{a+b>0} \bm{Y}_{ab}  \cr
                     &\leq 5\alpha \mu r \|\bm{E}\|_\infty + 3\alpha \mu r \|\bm{E}\|_\infty \cr
                     &\leq \frac{\mu r}{4n}\sigma_1^L, \KW{\alpha \lesssim\frac{1}{\mu r}}\numberthis\label{init:eq: L-L0 inf norm}
 \end{align*}
where the last step uses \eqref{eq:SminusS}  and the bound of $\alpha $ in Assumption \nameref{assume:Sparse}.

~\\
(iv) From the thresholding rule, we know that
\[
[\bm{S}_{0}]_{ij}=[\mathcal{T}_{\zeta_{0}} (\bm{S}+\bm{L}-\bm{L}_{0})]_{ij} =
\begin{cases}
\mathcal{T}_{\zeta_{0}} ([\bm{S}+\bm{L}-\bm{L}_{0}]_{ij}) & (i,j)\in\Omega \cr
\mathcal{T}_{\zeta_{0}} ([\bm{L}-\bm{L}_{0}]_{ij})        & (i,j)\in\Omega^c  \cr
\end{cases}.
\]
So (\ref{init:eq:eigenvalues bound 1}), (\ref{init:eq: L-L0 inf norm}) and $\zeta_{0}=\frac{\mu r}{2n}\lambda_1$ imply $[\bm{S}_{0}]_{ij}=0$ for all $(i,j)\in\Omega^c$, i.e., $supp(\bm{S}_{0}):=\Omega_{0}\subset \Omega$. Also, for any entries of $\bm{S}-\bm{S}_{0}$, there hold
\[
[\bm{S}-\bm{S}_{0}]_{ij} =
\begin{cases}
0 &  \cr
[\bm{L}_{0}-\bm{L}]_{ij} &   \cr
[\bm{S}]_{ij}         & \cr
\end{cases} \leq
\begin{cases}
0 &  \cr
\|\bm{L}-\bm{L}_{0}\|_\infty       &  \cr
\|\bm{L}-\bm{L}_{0}\|_\infty +\zeta_{0} & \cr
\end{cases} \leq
\begin{cases}
0 &  (i,j)\in \Omega^c   \cr
\frac{\mu r}{4n} \sigma_1^L       & (i,j)\in \Omega_0   \cr
\frac{\mu r}{n}  \sigma_1^L & (i,j)\in \Omega\backslash\Omega_0. \cr
\end{cases}
\]
Here the last inequality follows from (\ref{init:eq:eigenvalues bound 1}) which implies $\zeta_{0}=\frac{\mu r}{2n}\lambda_1\leq \frac{3\mu r}{4n}\sigma_1^L$.
Therefore, we have
\[
supp(\bm{S}_{0})\subset\Omega\quad \textnormal{and}\quad \|\bm{S}-\bm{S}_{0}\|_\infty\leq\frac{\mu r}{n}\sigma_1^L.
\]
The proof is compete by  noting \eqref{eq:norm:L-L0} and the above results. 
\end{proof}

\section{Discussion and Future Direction} \label{sec:discussion}



We have presented a highly efficient algorithm AccAltProj for robust principal component analysis. The algorithm is developed by introducing a novel subspace projection step before the SVD truncation, which reduces the per iteration computational complexity of the algorithm of alternating projections significantly. Theoretical recovery guarantee has been established for the new algorithm, while numerical simulations show that our algorithm is superior to other state-of-the-art algorithms. 

There are three lines of research for future work. Firstly, the theoretical number of the non-zero entries in a sparse matrix below which AccAltProj can achieve successful recovery is highly pessimistic compared with our numerical findings. This suggests the possibility of improving the theoretical result. Secondly, recovery stability of the proposed algorithm to additive noise will be investigated in the future. Finally, this paper focuses on the fully observed setting. The proposed algorithm might be similarly extended to the partially observed setting where only partial entries of a matrix are observed. It is also interesting to study the recovery guarantee of  the proposed algorithm under this partial observed setting. 



\bibliographystyle{unsrt}
\bibliography{ref}
\end{document}